%% file: csp_isomorphism.tex
\documentclass{article}
\usepackage{hyperref}
\usepackage{amsmath}
\usepackage{amssymb}
\usepackage{amsmath}
\usepackage{amsthm}
\usepackage{enumitem}
\usepackage{graphicx}
\usepackage{subcaption}
\usepackage{tikz}
\usepackage{tkz-graph}
\usepackage{biblatex}
\usepackage[margin=1in]{geometry}

\usetikzlibrary{shapes.geometric}

\newtheorem{theorem}{Theorem}
\newtheorem{proposition}{Proposition}

\newtheorem{lemma}{Lemma}
\newtheorem{corollary}{Corollary}
\newtheorem*{theorem*}{Theorem}

\theoremstyle{remark}
\newtheorem{remark}{Remark}

\theoremstyle{definition}
\newtheorem{definition}{Definition} % Use this for non-trivial definitions

\DeclareMathOperator{\holant}{Holant}
\DeclareMathOperator{\n}{\mathbb{N}}
\DeclareMathOperator{\pli}{\mathcal{PLI}}
\DeclareMathOperator{\plisimp}{\mathcal{PLI}^{\text{simp}}}
\DeclareMathOperator{\aut}{Aut}
\DeclareMathOperator{\fc}{\mathcal{F}}
\DeclareMathOperator{\gc}{\mathcal{G}}
\DeclareMathOperator{\vi}{\mathbf{i}}

\DeclareMathOperator{\vy}{\mathbf{y}}
\DeclareMathOperator{\vx}{\mathbf{x}}
\DeclareMathOperator{\vp}{\mathbf{p}}
\DeclareMathOperator{\vv}{\mathbf{v}}
\DeclareMathOperator{\jc}{\mathcal{J}}
\DeclareMathOperator{\gk}{\mathfrak{G}}
\DeclareMathOperator{\qk}{\mathfrak{Q}}
\DeclareMathOperator{\qb}{\mathbf{Q}}
\DeclareMathOperator{\e}{\mathbf{E}}
\DeclareMathOperator{\s}{\mathbf{S}}
\DeclareMathOperator{\f}{\mathbf{F}}
\DeclareMathOperator{\ff}{\mathbb{F}}
\DeclareMathOperator{\ii}{\mathbf{I}}
\renewcommand{\sb}{\text{S}}

\renewcommand{\k}{\mathbf{K}}
\DeclareMathOperator{\ext}{ext}
\DeclareMathOperator{\arity}{arity}
\newcommand{\tcwd}[1]{\langle #1 \rangle_{+, \circ, \otimes ,*}}
\newcommand{\tcwdn}[1]{\langle #1 \rangle_{\circ, \otimes ,*}}
\renewcommand{\c}{\mathbb{C}}
\renewcommand{\r}{\mathbb{R}}
\DeclareMathOperator{\eq}{\mathcal{EQ}}

\renewcommand*{\VertexSmallMinSize}{3pt}

\newcommand{\DEdge}[2]{
    \draw[thick] (#1) -- (#2) node[draw, fill=white, kite, kite vertex angles = 120, minimum size = 4pt, inner sep = 1pt, pos=0.15, sloped] {};
}

\title{Equality on all \#CSP Instances Yields Constraint Function Isomorphism via Interpolation and Intertwiners}
\author{Ben Young\footnote{Department of Computer Sciences, University of Wisconsin-Madison}
\\ \texttt{\href{mailto:benyoung@cs.wisc.edu}{benyoung@cs.wisc.edu}}}
\date{}

\begin{document}
\maketitle

\begin{abstract}
    A fundamental result in the study of graph homomorphisms is Lov\'asz's theorem
    \cite{lovasz_operations} that two graphs are isomorphic if and only if they admit the same number
    of homomorphisms from every graph.
    A line of work extending Lov\'asz's result to more general types of graphs was recently capped by
    Cai and Govorov \cite{homomorphism}, who showed that it holds for graphs with vertex and edge
    weights from an arbitrary field of characteristic 0. In this work, we generalize from graph
    homomorphism -- a special case of \#CSP with a single binary function -- to general \#CSP by showing
    that two sets $\fc$ and $\gc$ of arbitrary constraint functions are isomorphic if and only 
    if the partition function of any \#CSP instance is unchanged when we replace the functions in $\fc$ with
    those in $\gc$. We give two very different proofs of this result. First, we demonstrate the power
    of the simple Vandermonde interpolation technique used in \cite{homomorphism} by extending it
    to general \#CSP. Second, we give a proof using the intertwiners of the automorphism
    group of a constraint function set, a concept from the representation theory of compact groups.
    This proof is a generalization of a classical version of the recent proof of the Lov\'asz-type 
    result in \cite{planar} relating quantum isomorphism and homomorphisms from planar graphs.
\end{abstract}

\input{introduction}
\input{preliminaries}
\input{special_case}
\input{general_case}
\input{quantum_proof}

\section*{Acknowledgements}
The author thanks Jin-Yi Cai and Austen Fan for helpful discussions.

\printbibliography
\end{document}

%% file: introduction.tex
\section{Introduction}
\label{sec:intro}

\paragraph{Graph homomorphisms.}
A homomorphism from graph $K$ to graph $X$ is an adjacency-preserving map from $V(K)$ to $V(X)$.
Since graph homomorphisms' introduction in \cite{lovasz_operations}, counting the number of homomorphisms
from $K$ to $X$ has emerged as a well-studied problem in theoretical computer science and combinatorics.
The number of homomorphisms from graph $K$ to $X$, denoted $\hom(K,X)$, 
can be computed as the evaluation of a 
\emph{partition function}, parameterized by $X$, on $K$ -- the sum over all maps 
$\phi: V(K) \to V(X)$ of the product $\prod_{(u,v) \in E(K)} (A_X)_{\phi(u),\phi(v)}$, where
$A_X$ is the adjacency matrix of $X$.
The partition function perspective leads to extensions of graph homomorphism to more general types of graphs, as well as
a natural view of graph homomorphism as a special case of counting constraint satisfaction problems, or \#CSP.

One such more general type of graph has a real weight assigned to each edge and a nonnegative real weight
assigned to each vertex. In the partition function formulation of graph homomorphism to such a graph $X$,
we simply use the weighted adjacency matrix $A_X$, and add factors for the vertex weights.
The problem of counting homomorphisms to such weighted graphs was studied in 
\cite{freedman_reflection_2006, lovasz, lovasz_contractors_2009}. These works prove their results by
studying \emph{graph algebras} of formal $\c$-combinations of $k$-labeled graphs (called \emph{quantum
graphs}, which we generalize in discussion above \autoref{thm:intertwinersigmatrix}), using the $k$-labeled graph product
extended by our \autoref{def:klabeled} below. 
In \cite{lovasz}, Lov\'asz 
extended to these weighted graphs his result, proved forty years prior in \cite{lovasz_operations},
that two graphs are isomorphic if and
only if they admit the same number of homomorphisms from every graph. Throughout this paper, we will
refer to such generalizations of Lov\'asz's original isomorphism theorem as ``Lov\'asz-type results''.

Still making use of quantum and $k$-labeled graphs, but applying
invariant theory and the Nullstellensatz from algebraic geometry,
Schrijver \cite{schrijver_graph_2009} studies homomorphisms to graphs with complex edge weights but
without vertex weights, and proves a Lov\'asz-type result for such complex-edge-weighted graphs.
Using similar proof methods, Regts \cite{regts} studies ``vertex-coloring models'' -- 
homomorphisms to graphs with arbitrary vertex
and edge weights, provided that no nonempty subset of vertex weights sums to zero.

Finally, Cai and Govorov \cite{homomorphism} prove a Lov\'asz-type result for graphs with any vertex and edge
weights from an arbitrary field $\ff$ of characteristic 0 (as discussed in \cite{homomorphism}, a similar,
slightly weaker result can also be obtained from the results of \cite{goodall_matroid_2016}, proved
using $k$-labeled graph algebras and matroid invariants), and show this is the most general
possible Lov\'asz-type result for graph homomorphism. Cai and Govorov obtain this generality, overcoming the
algebraic approaches' technical difficulties of vertex weights summing to 0, by applying a simple, direct
\emph{Vandermonde interpolation} technique, dependent only on no algebraic results aside from the fact
that a Vandermonde matrix with distinct roots is nonsingular. It is remarkable that such a simple tool
unifies all previous Lov\'asz-type results, and in \autoref{sec:interpolation} we give a further 
demonstration of its power by using it to extend Cai and Govorov's results to \#CSP.

Another line of work studies homomorphisms \emph{from} restricted classes of graphs rather than
homomorphisms \emph{to} expanded classes as above, and
uses invariance of homomorphism counts from restricted classes of graphs to characterize 
relaxations of graph isomorphism. In \cite{dvorak_recognizing_2010}, using the techniques of
\cite{lovasz_contractors_2009}, Dvo{\v r}\'ak showed that
homomorphism count from 2-degenerate graphs suffices to determine a graph up to isomorphism, and that
homomorphism count from graphs of treewidth at most $k$ determines graphs up to their $k$-degree
refinements, but not up to isomorphism. Then in \cite{dell} it was shown that two graphs admit the same
number of homomorphisms from all graphs of treewidth at most $k$ if and only if they are
indistinguishable by the $k$-dimensional Weisfeiler-Leman algorithm.

Most notably, Man{\v c}inska and Roberson showed in \cite{planar} that two graphs are \emph{quantum
isomorphic} if and only if they admit the same number of homomorphisms from all planar graphs.
Quantum isomorphism is defined using \emph{quantum permutation groups}, and its characterization by 
planar graph homomorphisms is achieved using the \emph{intertwiner space} of the quantum automorphism
group of a graph, a quantum permutation group analogous to the graph's classical automorphism group.
A key component of the proof is a `quantum' version of Woronowicz's Tannaka-Krein duality
\cite{woronowicz_tannaka}, which implies that a quantum permutation group is uniquely determined by its
intertwiner space. A `classical' version of Tannaka-Krein duality (\autoref{thm:tannaka}) similarly
applies to the intertwiner space of the classical automorphism group of a graph, or, more generally, of a
set of \#CSP constraint functions. Using this, in \autoref{sec:quantum} we give a classical version of Man{\v c}inska and Roberson's proof, generalized
to sets of real-valued \#CSP functions -- the same result as proved via Vandermonde interpolation
in \autoref{sec:interpolation}, but restricted to $\mathbb{R}$ rather than general fields.
In an upcoming work, we also generalize Man{\v c}inska and Roberson's original `quantum' result
to \#CSP.

\paragraph{Counting complexity and \#CSP.} A \#CSP$(\fc)$ problem is parameterized by a set $\fc$ of 
$\ff$-valued constraint 
functions on one or more inputs from a finite domain $V(\fc)$. The problem input is a \#CSP instance, consisting of a
set of constraints, each applying a constraint function to a subset of variables. The output is the
value of the partition function, the sum over all variable assignments of the product of the constraint
evaluations. Letting $\fc = \{A_X\}$, $V(\fc) = V(X)$, the variable set be $V(K)$, and the constraint
set be $E(K)$ with each edge-constraint applying $A_X$ to its two endpoints, one can see from the
partition function formulation of graph homomorphism above that counting homomorphisms from
$K$ to $X$ is the special case of \#CSP$(\fc)$ on instance $K$ where $\fc$ contains a single binary (arity-2) 
constraint function $A_X$.

Counting graph homomorphisms is a central problem in counting complexity, both in its own
right and as a special case of \#CSP. Both settings have seen many significant dichotomy theorems
classifying the partition function as either tractable of \#P-hard to compute, depending on $X$ or $\fc$,
respectively. Graph homomorphism dichotomies were established for unweighted graphs in
\cite{dyer_complexity}, nonnegative-real-weighted graphs in \cite{bulatov_complexity_2005, 
cai_nonnegative}, real-weighted graphs in \cite{goldberg_complexity_2010}, and finally complex-weighted
graphs in \cite{cai2013graph}. For \#CSP, dichomies were established for sets of 0-1 valued
constraint functions in \cite{bulatov_2013, dyer_richerby}, nonnegative-real-valued constraint functions in \cite{cai-chen-lu}, and complex-valued constraint functions in \cite{cai-chen-complexity}.

Extending the notion of graph isomorphism, we say two constraint functions $F_1$ and $F_2$ of the same arity $n$
on the same domain $V(F)$ are isomorphic if there is a permutation $\sigma$ of $V(F)$ such that
$F_1(x_1,\ldots,x_n) = F_2(\sigma(x_1),\ldots,\sigma(x_n))$ for all $x_1,\ldots,x_n \in V(F)$.
Two sets of constraint functions $\fc$ and $\gc$ are isomorphic if there is a common 
isomorphism between each $F \in \fc$ and a corresponding $G \in \gc$. Some similar concepts exist:
in \cite{bohler2002equivalence,bohler2004complexity}, B{\"o}hler et al. study ``constraint isomorphism''
between Boolean \#CSP \emph{instances} (rather than constraint functions) that involves permuting
variables (rather than domain elements). One can also view an $n$-ary constraint function $F$ as a
tensor in $\ff^{V(F)^n}$; from this perspective the notion of tensor isomorphism in
\cite{grochow2021complexity} is a relaxation of constraint function isomorphism from permutations to invertible
linear transformations on each dimension.

\paragraph{Our results.}
Our main result is the following theorem, which to our knowledge is the first Lov\'asz-type result of
any kind for \#CSP.
\begin{theorem*}[\autoref{thm:mainresult}, informal]
    For field $\ff$ of characteristic 0, sets $\fc$ and $\gc$ of $\ff$-valued constraint functions are isomorphic if and only if the partition function
    of every \#CSP$(\fc)$ instance is preserved when we replace every constraint function in $\fc$ with
    the corresponding function in $\gc$.
\end{theorem*}
We prove \autoref{thm:mainresult} in the style of \cite{homomorphism} (Vandermonde interpolation) in
\autoref{sec:interpolation} and in the style of \cite{planar} (intertwiner spaces) in \autoref{sec:quantum}.
The former actually proves a more general result
(\autoref{thm:nowellbalanced}) applying to $k$-labeled \#CSP instances and constraint function sets
with domain/vertex weights.
By the above discussion, one can see that the Lov\'asz-type result, proved in
\cite{homomorphism}, that $\ff$-weighted graphs $X$ and $Y$ are isomorphic if and only if
$\hom(K,X) = \hom(K,Y)$ for all graphs $K$, is the special case of \autoref{thm:mainresult} where
$\fc = \{A_X\}$ and $\gc = \{A_Y\}$.
We carry out the latter proof in the Holant framework from counting complexity. While 
\cite{planar} does not explicitly make use of the Holant framework, we find it a very natural setting. 
Roughly, the main idea is to express the intertwiner space of $\aut(\fc)$ as the span of the \emph{signature matrices}
of Holant \emph{gadgets}, which one can view as \#CSP instances with free/input variables, via a
decomposition (\autoref{thm:generategk}) of any gadget into fundamental `building block' gadgets
and an analogous characterization (\autoref{thm:chassaniol}) of the intertwiner space of $\aut(\fc)$.
While the version of \autoref{thm:mainresult} in \autoref{sec:quantum} is restricted to constraint functions over $\r$, rather than 
over general fields as in \autoref{sec:interpolation}, we believe that the combinatorial reasoning
in \autoref{sec:quantum} is more intuitive than the interpolation technique
of \autoref{sec:interpolation} and \cite{homomorphism}, as well as the algebraic proofs of 
\cite{freedman_reflection_2006, lovasz, lovasz_contractors_2009, schrijver_graph_2009, goodall_matroid_2016} discussed above.
The intertwiner proof also demonstrates the surprisingly natural application of the powerful 
representation theoretic tools of intertwiner spaces and Tannaka-Krein duality to \#CSP and Holant
theory. We hope it inspires further applications of representation theory to theoretical computer
science.

%% file: preliminaries.tex
\section{Preliminaries}
\label{sec:preliminaries}

For notational brevity, following 
\cite{dudek} and others working with $n$-ary structures,
write $x_i^j$ to mean
$(x_i, \ldots, x_j)$ if $j \geq i$, and the empty list if $i > j$.
When the index range is clear, we simply write $\vx = (x_1, \ldots, x_r)$.
For any $q \in \mathbb{N}$, write $[q] = \{1,2,\ldots,q\}$ and $[0,q) = \{0,1,\ldots,q-1\}$.
For sets $A$ and $B$, $A^B$ denotes the set of functions from $B$ to $A$. For a set $B$ with an implicit
linear order and $a_b \in A$ for every $b \in B$, $(a_b)_{b \in B}$ denotes a tuple of elements of $A$ indexed and ordered by $B$. We will
view $(a_b)_{b \in B}$ as an element of $A^B$, and will abbreviate
it as simply $(a_b)$ if the index set $B$ is clear from context.
Let $\sb_q$ be the symmetric group of permutations on $[q]$.
Throughout, let $\ff$ be a field of characteristic 0.

\paragraph{Counting Constraint Satisfaction Problems.}
Any function $F: [q]^{n_F} \to \ff$ on $n_F \geq 1$ variables taking values in $[q]$ is a
\emph{constraint function} with \emph{domain} $[q]$ and \emph{arity} $n_F$. 
When $n_F = 2$, one can view $F$ as a $q \times q$ matrix with entries in $\ff$, the adjacency matrix
of an \emph{$\ff$-weighted graph}~\cite{homomorphism}.
Denote sets of constraint functions by calligraphic letters such as $\fc$ and $\gc$.
It is assumed that all constraint functions in a set $\fc$ have the same domain, denoted by $V(\fc)$
($V$ stands for `vertices', terminology inherited from the weighted graph special case) 
and that all constraint function sets are finite.

\begin{definition}[\#CSP, $Z_{\fc}$]
A \emph{\#CSP problem} \#CSP$(\fc)$ is parameterized by a set
$\fc$ of constraint functions.
A \emph{\#CSP$(\fc)$ instance} $K = (V,C)$ is defined by a set $V$ of variables and a
multiset $C$ of \emph{constraints}.
Each constraint $(F, v_{i_1},\ldots,v_{i_{n_F}})$ consists of a constraint
function $F \in \fc$ and
an ordered tuple of variables to which $F$ is applied. 

The \emph{partition function} $Z_{\fc}$, on input \#CSP$(\fc)$ instance $K = (V,C)$, outputs
\[
    Z_{\fc}(K) = \sum_{\phi: V \to V(\fc)} 
    \prod_{(F, v_{i_1},\ldots,v_{i_{n_F}}) \in C} 
    F(\phi(v_{i_1}),\ldots,\phi(v_{i_{n_F}})).
\]
\end{definition}

\begin{definition}[Compatible constraint function sets, $K_{\fc\to\gc}$]
Constraint function sets 
$\fc = \{F_j\}_{j\in [t_f]}$, $\gc = \{G_j\}_{j\in [t_g]}$ are \emph{compatible} if
$t_f = t_g = t$ and
$F_j: [q_f]^{n_j} \to \ff$ and $G_j: [q_g]^{n_j} \to \ff$
for all $j \in [t]$ (in other words, equal-indexed constraint functions have the same arity). 
We call $F_j$ and $G_j$ \emph{corresponding} constraint functions.

For compatible constraint function sets $\fc$ and $\gc$ and
any \#CSP$(\fc)$ instance $K$, define a \#CSP$(\gc)$ instance
$K_{\fc\to\gc}$ by replacing every constraint in $K$ with
the corresponding constraint
function in $\gc$ applied to the same variable tuple.
\end{definition}

More generally, we may define a domain-weighted \#CSP problem. 
\begin{definition}[\#CSP$(\fc, \alpha)$, $Z_{\fc, \alpha}$]
\label{def:csp}
The problem \#CSP$(\fc, \alpha)$ is parameterized by
a set $\fc$ of constraint functions with domain $V(\fc) = [q]$, and a vector of \emph{domain weights} 
$\alpha \in (\ff \setminus \{0\})^q$. The partition function
$Z_{\fc, \alpha}$, defined on \#CSP$(\fc)$ instances $K = (V,C)$ as
above, is
\[
    Z_{\fc,\alpha}(K) = 
    \sum_{\phi: V \to [q]} 
    \prod_{v \in V} \alpha_{\phi(v)}
    \prod_{(F, v_{i_1},\ldots,v_{i_{n_F}}) \in C} 
    F(\phi(v_{i_1}),\ldots,\phi(v_{i_{n_F}}))
\]
\end{definition}
In particular, $Z_{\fc,\mathbf{1}} = Z_{\fc}$, where $\mathbf{1}$ is the all-ones vector.
We use ``\#CSP$(\fc)$ instance'' and ``\#CSP$(\fc,\alpha)$ instance''
interchangably, since a \#CSP$(\fc,\alpha)$ instance does not depend on the domain weights -- it is
identical to a \#CSP$(\fc)$ instance.

\begin{definition}[$k$-labeled \#CSP instance (product), {$\pli[\fc;k]$}, {$\plisimp[\fc;k]$}]
    \label{def:klabeled}
    A \#CSP instance $K = (V,C)$ is \emph{$k$-labeled} if $k$ variables are labeled by $1,2,\ldots,k$. A
    single variable cannot be labeled more than once.
    Define the \emph{product} $K_1K_2$ of two $k$-labeled \#CSP$(\fc)$ instances $K_1 = (V_1,C_1), K_2 = (V_2,C_2)$ as follows. 
    For $i \in [k]$, let $u_i \in V_1, v_i \in V_2$ be the variables labeled $i$ in $V_1$ and $V_2$,
    respectively. Define a new variable set $V$ by starting with $V_1 \sqcup V_2$, then for each $i \in [k]$
    merging $u_i$ and $v_i$ into a new variable $w_i$, and label $w_i$ by $i$.
    Then define a new constraint multiset $C$ by
    starting with $C_1 \cup C_2$ (multiset union), then for every $i \in [q]$ replacing every occurrence of $u_i$ or $v_i$ in each constraint with $w_i$. Then take $K_1K_2 = (V,C)$.

    Define $\pli[\fc;k]$ to be the set of $k$-labeled \#CSP$(\fc)$ instances.
    Let $U_k = (V, \varnothing) \in \pli[\fc;k]$, where $V$ contains exactly $k$ vertices labeled $1,\ldots,k$.
    The $k$-labeled instance product is commutative and associative and has identity $U_k$, so 
    $\pli[\fc;k]$ forms a commutative monoid under this product. Let $\plisimp[\fc;k]$ denote the submonoid
    of $\pli[\fc;k]$ consisting of \emph{simple} instances -- those where the variables in any constraint
    $c \in C$ are distinct, the multiplicity of every constraint in $C$ is 1 up to permutation of the
    order of its variables, and no constraint contains only labeled variables.
\end{definition}
Observe that, for $K_1,K_2 \in \pli[\fc;k]$, $(K_1)_{\fc\to\gc}(K_2)_{\fc\to\gc} = (K_1K_2)_{\fc\to\gc} \in \pli[\gc;k]$.

\begin{definition}[$Z_{\fc, \alpha}^{\psi}$]
For $K = (V,C) \in \pli[\fc;k]$ and a map $\psi: [k] \to [q]$ fixing, or \emph{pinning}, the values of the labeled
variables, define
\[
    Z_{\fc, \alpha}^{\psi}(K) = \sum_{\phi: V \to [q] \text{ extends } \psi} 
    \frac{\alpha_{\phi}}{\alpha_{\psi}} 
    \prod_{(F, v_{i_1},\ldots,v_{i_{n_F}}) \in C} 
    F(\phi(v_{i_1}),\ldots,\phi(v_{i_{n_F}})),
\]
where
\[
    \alpha_{\phi} = \prod_{v \in V} \alpha_{\phi(v)}
    \text{ and }
    \alpha_{\psi} = \prod_{i \in [k]} \alpha_{\psi(i)},
\]
and $\phi$ extends $\psi$ means $\phi$ assigns value $\psi(i)$ to the variable labeled $i$. Then
we have
\[
    Z_{H, \alpha}(K) = \sum_{\psi: [k] \to [q]} \alpha_{\psi} Z_{H,\alpha}^{\psi}(K).
\]
\end{definition}

The $k$-labeled instance product $K_1K_2$ merges the labeled variables, and the unlabeled variables of
$K_1$ and $K_2$ both still appear in constraints from $K_1$ and $K_2$ with the combined 
labeled variables. The unlabeled
variables of $K_1$ take values independently of the unlabeled variables of $K_2$ (i.e.
they appear in no constraints with each other). Hence
\begin{equation}
    \label{eq:mult}
    Z_{\fc, \alpha}^{\psi} (K_1K_2) = Z_{\fc, \alpha}^{\psi} (K_1)Z_{\fc, \alpha}^{\psi} (K_2).
\end{equation}

For fixed $\fc = \{F_1,\ldots,F_t\}$ with common domain $[q]$, let 
\begin{equation}
    \label{eq:jc}
    \jc(\fc) = \{(j, \vx, r) \mid j \in [t], \vx \in [q]^{n_j-1}, r \in [n_j]\}
\end{equation}
(recall that $n_j$ is the arity of $F_j \in \fc$). 
If $n_j = 1$ ($F_j$ is unary), then say $[q]^{n_j-1} = [q]^0 = \{()\}$ (the set containing the
empty tuple), so $\vx \in [q]^{n_j - 1}$ means $\vx = ()$.
$\mathcal{J}(\fc)$ represents all `configurations'
in which we may fill in the remaining arguments of an application a function in $\fc$ when given a
single distinguished argument. Note that the length of $\vx$ and the domain of $r$ both depend on $j$
(the choice of $F_j \in \fc$).
Domain elements $i,i' \in [q]$ are \emph{twins} if 
\[
    F_j(x_1^{r-1},i,x_r^{n_j-1}) = F_j(x_1^{r-1},i',x_r^{n_j-1}) \text{ for every } 
    (j, \vx, r) \in \jc(\fc).
\]
If $n_j = 1$ and $\vx = ()$, then $F_j(x_1^{r-1},i,x_r^{n_j-1}) = F_j(i)$.
If every $F \in \fc$ is \emph{symmetric}, meaning $F$ is invariant
under permutations of the order of its inputs, then say $\fc$ is \emph{symmetric}, and $i,i' \in [q]$ are
twins if $F_j(i, \vx) = F_j(i', \vx)$ for every $j \in [t]$ and $\vx \in [q]^{n_j-1}$,
where we abbreviate $F_j(i, \vx) = F_j(i, x_1^{n_j-1})$. If $n_j = 1$ and $\vx = ()$, then
$F_j(i,\vx) = F_j(i)$.
$\fc$ is \emph{twin-free} if no two domain elements are twins. Equivalently, $\fc$ is twin-free iff
the tuples
\[
    \left(F_j(x_1^{r-1},i,x_r^{n_j-1})\right)_{(j,\vx,r) \in \jc(\fc)}
\]
are pairwise distinct for $i \in [q]$.
If $\fc$ is symmetric,
then $\fc$ is twin free iff the tuples 
$\left(F_j(i,\vx)\right)_{j \in [t], \vx \in [q]^{n_j-1}}$ are pairwise distinct for $i \in [q]$.

For any $\fc$ and $\alpha$, let $I_1,\ldots,I_s$ be the partition of $[q]$ under
the twin relation. Define the twin-contracted constraint function set $\widetilde{\fc}$ with
domain $[s]$ and, for each $F \in \fc$ and $\vy \in [s]^{n_F}$, $\widetilde{F}(\vy) := F(\vx)$ for arbitrary $x_i \in I_{y_i}$, $i \in [n_F]$.
Define the twin-contracted domain weights $\widetilde{\alpha}$ by
$\widetilde{\alpha}_{\ell} = \sum_{j \in I_{\ell}} \alpha_j$, $\ell \in [s]$. $\widetilde{\fc}$ is now twin-free and
\[
    Z_{\widetilde{\fc},\widetilde{\alpha}}(K_{\fc \to \widetilde{\fc}}) = Z_{\fc, \alpha}(K)
\]
for every \#CSP$(\fc)$ instance $K$.

\begin{definition}[$\cong$, $\aut(\fc,\alpha)$]
    \label{def:iso}
For $(F,\alpha)$ and $(G,\beta)$ with domain $[q]$ and arity $n$, a permutation
$\sigma \in \sb_q$ is a \emph{domain-weighted isomorphism} from $(F,\alpha)$ to $(G,\beta)$ if
$F(\vx) = G(\sigma(\vx))$ for all $\vx \in [q]^n$ (where $\sigma(\vx) = (\sigma(x_1),\ldots,\sigma(x_n))$) and
$\alpha_i = \beta_{\sigma(i)}$ for all $i \in [q]$. 

For compatible constraint function sets $\fc$ and $\gc$ of cardinality $t$ and common domain $[q]$,
$(\fc, \alpha)$ and $(\gc, \beta)$ are
isomorphic ($(\fc, \alpha) \cong (\gc, \beta)$) if there is a $\sigma \in S_q$
which is an isomorphism between $(F_j, \alpha)$ and $(G_j, \beta)$ for all $j \in [t]$.

$\aut(\fc,\alpha)$ is the set of all domain-weighted
isomorphisms from $(\fc,\alpha)$ to itself. 
\end{definition}
We emphasize that every corresponding pair of functions
in $\fc$ and $\gc$ must be isomorphic via the same $\sigma$.
Since a domain-weighted isomorphism is just a relabeling of the domain elements,
if $\varphi: [k] \to [q]$ and $\psi: [k] \to [q]$ satisfy $\psi = \sigma \circ \varphi$ for some
domain-weighted isomorphism $\sigma$ from $(\fc,\alpha)$ to $(\gc,\beta)$, then
$Z_{\fc,\alpha}^{\varphi}(K) = Z_{\gc,\beta}^{\psi}(K_{\fc\to\gc})$ for every $K \in \pli[\fc;k]$.
In this work, we aim to prove the converse of this fact.

Let $X$ be a $\ff$-weighted graph with adjacency matrix $A_X \in \ff^{q \times q}$ and vertex weights
$\alpha \in \ff^q$.
To compute $\text{hom}(K,X)$,
construct a \#CSP$(\{A_X\},\alpha)$ instance $K'$ with variable set $V(K)$ and each edge of $K$ is a constraint
applying function $F = A_X: [q]^2 \to \ff$ to the edge's two endpoints.
The variables take values in $V(X)$, so a variable assignment $\phi$ is a mapping of $K$'s vertices to 
$X$'s. Hence $Z_{A_X, \alpha}(K') = \text{hom}(G,K)$.
For $\fc = \{F\}, \gc = \{G\}$ for binary $F$ and $G$, the constructions in Definitions 
\ref{def:csp}-\ref{def:iso} are equivalent to the corresponding special cases in \cite{homomorphism}.
In particular, $k$-labeled \#CSP instances are a generalization of $k$-labeled graphs.

We now have the notation to state our main theorem.
\begin{theorem}
    \label{thm:mainresult}
    Let $\ff$ be a field of characteristic 0, and
    let $\fc$ and $\gc$ be compatible $\ff$-valued constraint function sets. 
    Then $\fc \cong \gc$ if and only if
    $Z_{\fc}(K) = Z_{\gc}(K_{\fc\to\gc})$
    for every \#CSP$(\fc)$ instance $K$.
\end{theorem}

\paragraph{Vandermonde Interpolation.}
Next, we introduce the useful Vandermonde interpolation technique from \cite{homomorphism}, which is
essentially the only technique used to prove our main result. The basis for the technique is the
following simple lemma.

\begin{lemma}[{\cite[{Lemma 4.1}]{homomorphism}}]
    \label{lem:firstvandermonde}
    Let $n \geq 0$ and $a_i,x_i \in \ff$ for $1 \leq i \leq n$, and suppose $\sum_{i=1}^n a_i x_i^j = 0$
    for all $0 \leq j < n$. Then, for any function $f: \ff \to \ff$, we have 
    $\sum_{i=1}^n a_i f(x_i) = 0$.
\end{lemma}
The next result follows from iterated applications of \autoref{lem:firstvandermonde}.
\begin{corollary}[{\cite[{Corollary 4.2}]{homomorphism}}]
    Let $I$ and $J$ be finite index sets, and $a_{i}$, $b_{i,j} \in \ff$ for all $i \in I$,
    $j \in J$. Further, let $I = \bigsqcup_{\ell \in [s]} I_{\ell}$ be the partition of $I$ into
    equivalence classes defined by relation $\sim$, where $i \sim i'$ iff $b_{i,j} = b_{i',j}$ for all $j \in J$.
    If $\sum_{i \in I} a_i \prod_{j \in J} b_{i,j}^{p_j} = 0$, for all choices of
    $(p_j)_{j \in J}$ where each $0 \leq p_j < |I|$, then $\sum_{i \in I_{\ell}} a_i = 0$
    for every $\ell \in [s]$.
    \label{lem:vandermonde}
\end{corollary}
$I$ (and $J$) will often be the set of all $m$-tuples whose entries range over $[q]$, for some $m$
and $q$, and the product of $b_{ij}^{p_j}$s for a fixed tuple will have $i$ range over all the tuple's 
entries, rather than refer to the tuple itself.
In this case, we have the following corollary, used implicitly in \cite{homomorphism}.
\begin{corollary}
    \label{cor:vandermonde}
    Let $J$ be a finite index set and $q, m \geq 1$. Let $a_{\vi} \in \ff$ for $\vi \in [q]^m$ and $b_{i,j} \in \ff$ for 
    $i \in [q], j \in J$, and for $i, i' \in [q]$, say $i \sim i'$ iff $b_{i,j} = b_{i',j}$ for all
    $j \in J$. Let $[q]^m = \bigsqcup_{\ell \in [s]} I_{\ell}$ be a partition of
    $[q]^m$ into equivalence classes defined by relation $\approx$, where $\vi \approx \vi'$ if 
    $i_h \sim i'_h$ for all $h \in [m]$. If 
    \[
        \sum_{\vi \in [q]^m} a_{\vi} \prod_{j \in J, h \in [m]} b_{i_h, j}^{p_{h,j}} = 0
    \]
    for every choice of $(p_{h,j})_{h \in [m], j \in J}$ where each $0 \leq p_{h,j} < q$, then
    $\sum_{\vi \in I_{\ell}} a_{\vi} = 0$ for every $\ell \in [s]$.
\end{corollary}
\begin{proof}
    Separating the sum over $i_m$, which we rename to $i$, we have
    \begin{equation}
        \label{eq:extractim}
        \sum_{i \in [q]} \left( \sum_{i_1^{m-1} \in [q]} a_{\vi} \prod_{j \in J, h \in [m-1]}b_{i_h, j}^{p_{h,j}} \right)
        \left( \prod_{j \in J} b_{i, j}^{p_{m,j}}\right) = 0.
    \end{equation}
    Applying \autoref{lem:vandermonde} with 
    \[
        I := [q],\quad
        a_{i} := \sum_{i_1^{m-1} \in [q]} a_{\vi} \prod_{j \in J, h \in [m-1]}b_{i_h, j}^{p_{h,j}}
        \text{ for $i \in [q]$, and } p_{j} := p_{m,j},
    \]
    we obtain
    \begin{equation}
        \sum_{i \in I'_{\ell_1}} \left(\sum_{i_1^{m-1} \in [q]} a_{\vi} \prod_{j \in J, h \in [m-1]}b_{i_h, j}^{p_{h,j}}\right) = 0.
        \label{eq:peelone}
    \end{equation}
    for every $\ell_1 \in [s']$, where $[q] = \bigsqcup_{\ell_1 \in [s']} I'_{\ell}$ is a partition of $[q]$ into the equivalence classes of $\sim$.
    Renaming $i_{m-1}$ to $i$, the LHS of \eqref{eq:peelone} is equal to
    \[
        \sum_{i \in [q]} \left( \sum_{i_1^{m-2} \in [q]} 
            \left( \sum_{i_m \in I'_{\ell_1}} a_{\vi} \right) 
        \prod_{j \in J, h \in [m-2]}b_{i_h, j}^{p_{h,j}}\right) 
        \left( \prod_{j \in J} b_{i_{m-1}, j}^{p_{m-1,j}}\right),
    \]
    which has a similar form to \eqref{eq:extractim}, but with the $m$th index removed.
    After $m$ repetitions, we eliminate the outer sum and both
    products and obtain the result.
\end{proof}

%% file: special_case.tex
\section{The Interpolation Proof}
\label{sec:interpolation}
\subsection{The Symmetric Ternary Case}
\label{sec:symm}

For clarity of exposition, we first prove the special case where all constraint functions are symmetric
and ternary. The general proof requires more sophisticated indexing but is not fundamentally different from the following proof of this special case.
\begin{proposition}
    \label{prop:balanced}
    Let $\fc = \{F_j \mid j \in [t]\}$ and
    $\gc = \{G_j \mid j \in [t]\}$ be compatible constraint function sets with domains $[q_f]$ and $[q_g]$,
    with $q_f \geq q_g$,
    such that every $F \in \fc$ and $G \in \gc$ are symmetric and have arity 3, and assume $\fc$ is twin-free.
    Let $\alpha, \beta$ be the domain weights associated with $\fc$ and $\gc$, respectively.

    Let $\varphi: [2k] \to [q_f]$ and $\psi: [2k] \to [q_g]$ for $k \geq 0$, and 
    for every $x,y \in [q_f]$, let
    \[
        I_{xy} = \{a \in [k] \mid \varphi(a) = x \wedge \varphi(a + k) = y\}.
    \]
    Assume $\varphi$ is
    \emph{well-balanced} -- that is, for every $x, y \in [q_f]$, $|I_{xy}| \geq 2q_f^3$.
    If $Z_{\fc, \alpha}^{\varphi}(K) = Z_{\gc,\beta}^{\psi}(K_{\fc\to\gc})$ for every 
    $K \in \plisimp[\fc;2k]$, then $q_f = q_g = q$ and there
    is a domain-weighted isomorphism $\sigma: [q] \to [q]$ from $(\fc,\alpha)$ to $(\gc,\beta)$ such that
    $\psi = \sigma \circ \varphi$.
\end{proposition}
\begin{proof}
    Consider $2k$-labeled variable set $V_1 = \{v, u_1,\ldots, u_{2k}\}$, where each $u_{\ell}$ is labeled $\ell$.
For a matrix $\chi \in \{0,1\}^{k \times t}$, define the following set of constraints on $V_1$:
\[
    C_{\chi} = \{(F_j, v, u_a, u_{a+k}) \mid a \in [k], j \in [t], \chi(a,j) = 1\},
\]
and define $2k$-labeled \#CSP$(\fc)$ instance $K_{\chi} = (V_1, C_{\chi}) \in \plisimp[\fc;2k]$.

Now construct a certain family of $\chi$. 
$|I_{xy}| \geq 2q_f^3 \geq 2q_fq_g^2$ for every $x, y \in [q_f]$, so, by the pigeonhole principle,
there is a function
$s: [q_f]^2 \to [q_g]^2$ such that for every $x, y \in [q_f]$, there exists a
$J_{xy} \subseteq I_{xy}$ such that $|J_{xy}| \geq 2q_f$ and 
for every $a \in J_{xy}$, $(\psi(a), \psi(a+k)) = s(x,y)$. 
For a choice of $(p_{xyj})_{x,y \in [q_f], j \in [t]} \in [0,2q_f)^{[q_f]^2 \times [t]}$,
for every $x,y \in [q_f]$ choose an arbitrary $P_{xyj} \subset J_{xy}$ of cardinality $p_{xyj}$
for every $j \in [t]$
and define $\chi = \chi((P_{xyj})_{x,y \in [q_f], j \in [t]})$ by
$\chi(a,j) = 1$ for all $x, y \in [q_f], j \in [t]$ and $a \in P_{xyj}$. Set the remaining entries of $\chi$ to $0$.
Let $R$ be the set of all such matrices $\chi$ for all choices of $(p_{xyj})
\in [0,2q_f)^{[q_f]^2 \times [t]}$.

To recap, for $j \in [t]$, if $a \in P_{xyj} \subset J_{xy} \subset I_{xy} \subset [k]$
for $x,y \in [q_f]$, 
then $(\varphi(a), \varphi(a+k)) = (x,y)$ and $(\psi(a),\psi(a+k)) = s(x,y)$.
Hence the variables $(u_a,u_{a+k})$ take values $(x,y)$ and $s(x,y)$ under $\varphi$ and $\psi$,
respectively. These values are independent of the choice of $a$ within $P_{xyj}$.
By construction, $K_{\chi}$ contains a constraint 
$(F_j,v,u_a,u_{a+k})$ for every $a \in P_{xyj}$.
Therefore $Z_{\fc, \alpha}^{\varphi}(K_{\chi}) = Z_{\gc,\beta}^{\psi}((K_{\chi})_{\fc\to\gc})$ for every $\chi \in R$
is equivalent to: for all
$(p_{xyj}) \in [0,2q_f)^{[q_f]^2 \times [t]}$,
\[
    \sum_{i=1}^{q_f} \alpha_i \prod_{x,y \in [q_f], j \in [t]} F_{j}(i,x,y)^{p_{xyj}}
    =
    \sum_{i=1}^{q_g} \beta_i \prod_{x,y \in [q_f], j \in [t]} G_{j}(i,s(x,y))^{p_{xyj}},
\]
where we write $G_{j}(i,s(x,y))$ to mean $G_j(i,s(x,y)_1,s(x,y)_2)$.
The sum over $i$ corresponds to the choice of assignment for the only free variable $v$.
Subtracting the RHS, we are left with a sum of $q_f + q_g \leq 2q_f$ terms on the LHS. 
Treating $[q_f]$ and $[q_g]$ as disjoint, apply \autoref{lem:vandermonde} to this sum with 
\[
    I := [q_f] \sqcup [q_g],
    ~J := [q_f]^2 \times [t],
    ~a_i := \begin{cases} \alpha_i & i \in [q_f] \\ \beta_i & i \in [q_g] \end{cases},
    ~b_{i,xyj} := \begin{cases} F_{j}(i,x,y) & i \in [q_f] \\
    G_{j}(i,s(x,y)) & i \in [q_g] \end{cases}.
\]

$\fc$ is twin-free, so the tuples $(F_j(i,x,y))_{x,y \in [q_f], j \in [t]}$ are pairwise
distinct for $i \in [q_f]$. Hence no equivalence class $I_{\ell}$ contains more than one element of
$[q_f]$. However, every $\alpha_i \neq 0$ by definition, so no equivalence class contains only a single
element of $[q_f]$. Thus there is a function $\sigma: [q_f]\to[q_g]$ such that
every $i \in [q_f]$ is in an equivalence class with $\sigma(i)\in [q_g]$ -- that is
\begin{equation}
    (F_j(i,x,y))_{x,y \in [q_f], j \in [t]} = (G_{j}(\sigma(i),s(x,y)))_{x,y \in [q_f],j\in[t]}
    \text{ for } i \in [q].
    \label{eq:beta}
\end{equation}
Since no two elements of $[q_f]$ are in the same equivalence class, $\sigma$ is injective, hence
bijective, as $q_f \geq q_g$. Thus $q_f = q_g = q$, and we view $\sigma$ as a function $[q]\to[q]$.
\autoref{lem:vandermonde} then gives
\begin{equation}
    \alpha_i = \beta_{\sigma(i)} \text{ for } i \in [q].
    \label{eq:alpha}
\end{equation}

Next, define another family of \#CSP$(\fc)$ instances. Fix $F \in \fc$ and corresponding $G \in \gc$.
Define a new $2k$-labeled variable set
\[
    V_2 = \{v, v', v'', u_1,\ldots,u_{2k}\},
\]
where each $u_{\ell}$ is
labeled $\ell$ (equivalent to the previous $V_1$, but with two new free variables $v'$ and $v''$).
For $\chi, \chi', \chi'' \in R$, define the following set of constraints on $V_2$:
\begin{align*}
    C_{\chi, \chi', \chi''} = \{(F, v, v', v'')\} &\cup \{(F_j, v, u_a, u_{a+k}) \mid \chi(a,j) = 1\} \\
                                                  &\cup \{(F_j, v', u_a, u_{a+k}) \mid \chi'(a,j) = 1\}\\
                                                  &\cup \{(F_j, v'', u_a, u_{a+k}) \mid \chi''(a,j) = 1\},
\end{align*}
and define $K_{\chi, \chi', \chi''} = (V_2, C_{\chi, \chi', \chi''})\in \plisimp[\fc;2k]$.
Every $(\chi, \chi', \chi'') \in R^3$ corresponds to three sequences of subsets
$(P_{xyj}), (P'_{xyj}), (P''_{xyj})$ and three sequences of
integers
$(p_{xyj}), (p'_{xyj}), (p''_{xyj}) \in [0,2q)^{[q]^2 \times [t]}$, where $p_{xyj}, p'_{xyj}, p''_{xyj}$ are
the cardinalities of
$P_{xyj}, P'_{xyj}, P''_{xyj}$, respectively, $P_{xyj}, P'_{xyj}, P''_{xyj} \subset J_{xy} \subset I_{xy}$, and
$\chi, \chi', \chi''$ are 1 at entry $(a,j)$ for $x,y \in [q], j \in [t]$, $a \in P_{xyj}, P'_{xyj}, P''_{xyj}$, respectively, and are 0 elsewhere.

Now the assumption $Z_{\fc, \alpha}^{\varphi}(K_{\chi, \chi', \chi''}) = Z_{\gc,\beta}^{\psi}((K_{\chi, \chi', \chi''})_{\fc\to\gc})$ for every 
$(\chi, \chi', \chi'') \in R^3$
is equivalent to: for all $(p_{xyj}), (p'_{xyj}), (p''_{xyj}) \in [0,2q)^{[q]^2 \times [t]}$,
\begin{align*}
    \sum_{i,i',i''=1}^q \alpha_i \alpha_{i'} \alpha_{i''} F(i,i',i'')
    &\prod_{x,y \in [q], j \in [t]}
    F_j(i,x,y)^{p_{xyj}}
    F_j(i',x,y)^{p'_{xyj}}
    F_j(i'',x,y)^{p''_{xyj}}\\
    = \sum_{i,i',i''=1}^q \beta_i\beta_{i'} \beta_{i''} G(i,i',i'')
    &\prod_{x,y \in [q], j \in [t]}^q 
    G_j(i,s(x,y))^{p_{xyj}}
    G_j(i',s(x,y))^{p'_{xyj}}
    G_j(i'',s(x,y))^{p''_{xyj}}.
\end{align*}
Subtracting the RHS and applying \eqref{eq:alpha} and \eqref{eq:beta} gives
\[
    \sum_{i,i',i'' = 1}^q \alpha_i\alpha_{i'}\alpha_{i''}(F(i,i',i'') - G(\sigma(i),\sigma(i'),\sigma(i'')))
    \prod_{x,y,j}
    F_j(i,x,y)^{p_{xyj}}
    F_j(i',x,y)^{p'_{xyj}}
    F_j(i'',x,y)^{p''_{xyj}} 
    = 0.
\]
$(F_j(i,x,y))_{xyj}$ are pairwise distinct for $i \in [q]$,
so the tuples $(F_j(i,x,y), F_j(i',x,y), F_j(i'',x,y))_{xyj}$
are distinct for distinct $(i,i',i'')$. Applying \autoref{cor:vandermonde} with 
\begin{align*}
    &m := 3,~ J = [q]^2 \times [t],~
    a_{i,i',i''} := \alpha_i\alpha_{i'}\alpha_{i''}(F(i,i',i'') - G(\sigma(i),\sigma(i'),\sigma(i''))),
    \\
    &b_{i,xyj} := F_j(i,x,y),~ p_{1,xyj} := p_{xyj},~p_{2,xyj} := p'_{xyj},~p_{3,xyj} := p''_{xyj},
\end{align*}
we obtain $\alpha_i\alpha_{i'}\alpha_{i''}(F(i,i',i'') - G(\sigma(i),\sigma(i'),\sigma(i''))) = 0$ for all $i,i',i''$ (each $p_{xyj}$ ranges over $[0,2q) \supset [0,q)$). Since each $\alpha_i \neq 0$ and our choice of $F$ and $G$
was arbitrary, this implies
\begin{equation}
    F(i,i',i'') =  G(\sigma(i),\sigma(i'),\sigma(i'')) \text{ for every } i,i',i'' \in [q] \text{ and
        every corresponding } F \in \fc, G \in \gc.
    \label{eq:sigma}
\end{equation}
Combined with \eqref{eq:alpha}, \eqref{eq:sigma} implies that $\sigma$ is a domain-weighted isomorphism 
between $(\fc,\alpha)$ and $(\gc,\beta)$.
Since $\fc$ is twin-free, \eqref{eq:sigma} also implies $\gc$ is also twin-free.

It remains to show that $\psi = \sigma \circ \varphi$. Again let $F \in \fc$ and $G \in \gc$ be
corresponding constraint functions.
Define a third family of \#CSP$(\fc)$ instances. Fix $c \in [2k]$.
Define a $2k$-labeled variable set
\[
    V_3 = \{v, v', u_1,\ldots,u_{2k}\},
\]
where each $u_{\ell}$ is labeled $\ell$
(equivalent to the previous $V_2$, but we have removed the free variable $v''$).
% Consider $\chi, \chi' \in R^2$, 
% but when choosing each $P_{xyj} \subseteq J_{xy}$ to define $\chi$ and each
% $P'_{xyj} \subseteq J_{xy}$ to define $\chi'$, choose these subsets so that
% $c \in I_{\widehat{x}\widehat{y}} \setminus (P_{\widehat{x}\widehat{y},j} \cup P'_{\widehat{x}\widehat{y},j})$ for every $j \in [t]$ (so $\chi_{c,j} = \chi'_{c,j} = 0$). This is always possible because
% $p_{\widehat{x}\widehat{y},j}, p'_{\widehat{x}\widehat{y},j} < 2q$ and $|I_{\widehat{x}\widehat{y}}| \geq 4m^3 \geq 4m$.
For $(\chi, \chi') \in R^2$, define the following set of constraints on $V_3$:
\[
    C_{\chi, \chi'} = \{(F, u_c, v, v')\} \cup \{(F_j, v, u_a, u_{a+k}) \mid \chi(a,j) = 1\}
    \cup \{(F_j, v', u_a, u_{a+k}) \mid \chi'(a,j) = 1\},
\]
and define $K_{\chi, \chi'} = (V_3, C_{\chi, \chi'}) \in \plisimp[\fc;2k]$.
Now $Z_{\fc, \alpha}^{\varphi}(K_{\chi, \chi'}) = Z_{\gc,\beta}^{\psi}((K_{\chi, \chi'})_{\fc\to\gc})$
for every $(\chi, \chi') \in R^2$
is equivalent to: for all $(p_{xyj}), (p'_{xyj}) \in [0,2q)^{[q]^2 \times [t]}$,
\begin{align*}
    &\sum_{i,i'=1}^q \alpha_i \alpha_{i'} F(\varphi(c),i,i')
    \prod_{x,y,j}
    F_j(i,x,y)^{p_{xyj}}
    F_j(i',x,y)^{p'_{xyj}}\\
    &= \sum_{i,i'=1}^q \beta_i \beta_{i'} G(\psi(c),i,i')
    \prod_{x,y,j}
    G_j(i,s(x,y))^{p_{xyj}}
    G_j(i',s(x,y))^{p'_{xyj}}.
\end{align*}
Subtracting the RHS and applying \eqref{eq:alpha} and \eqref{eq:beta} gives
\[
    \sum_{i,i'=1}^q \alpha_i \alpha_{i'} (F(\varphi(c),i,i') - G(\psi(c),\sigma(i),\sigma(i')))
    \prod_{x,y,j}
    F_j(i,x,y)^{p_{xyj}}
    F_j(i',x,y)^{p'_{xyj}} = 0.
\]
As above, the tuples $(F_j(i,x,y), F_j(i',x,y))_{xyj}$
are distinct for distinct $(i,i')$, so by a similar application of \autoref{cor:vandermonde} with
$m := 2$, 
we have $F(\varphi(c),i,i') = G(\psi(c),\sigma(i),\sigma(i'))$ for all
$i,i' \in [q]$. This holds for any corresponding pair $F \in \fc$ and $G \in \gc$, so, by
\eqref{eq:sigma},
\[
    G_j(\sigma(\varphi(c)),\sigma(i),\sigma(i')) = F_j(\varphi(c),i,i') = G_j(\psi(c),\sigma(i),\sigma(i'))
\]
for all $i,i' \in [q]$ and $j \in [t]$.
Since $\gc$ is twin-free and $\sigma$ is a bijection, this gives $\sigma(\varphi(c)) = \psi(c)$.
We chose $c \in [2k]$ arbitrarily, so $\psi = \sigma \circ \varphi$.
\end{proof}

%% file: general_case.tex
\subsection{The General Case}
\label{sec:general}

We now extend \autoref{prop:balanced} to general sets of arbitrary arity, non-necessarily-symmetric constraint functions, containing at least one non-unary constraint function.
\begin{lemma}
    \label{lem:balanced}
    Let $\fc = \{F_j \mid j \in [t]\}$ and
    $\gc = \{G_j \mid j \in [t]\}$ be compatible constraint function sets with domains $[q_f]$ and $[q_g]$,
    with $q_f \geq q_g$ and
    assume $\fc$ is twin-free.
    Let $\alpha, \beta$ be the domain weights associated with $\fc$ and $\gc$, respectively.
    Let $n$ be the maximum arity among all functions in $\fc$, and assume $n \geq 2$.
    Suppose $\varphi: [(n-1)k] \to [q_f]$ and $\psi: [(n-1)k] \to [q_g]$ for $k \geq 0$, and 
    for every $\vx \in [q_f]^{n-1}$, let
    \[
        I_{\vx} = \{a \in [k] \mid \varphi(a + (d-1)k) = x_d \text{ for all } d \in [n-1]\}.
    \]
    Assume $\varphi$ is
    \emph{well-balanced} -- that is, for every $\vx \in [q_f]^{n-1}$, $|I_{\vx}| \geq 2nq_f^n$.
    If $Z_{\fc, \alpha}^{\varphi}(K) = Z_{\gc,\beta}^{\psi}(K_{\fc\to\gc})$ for every 
    $K \in \plisimp[\fc;(n-1)k]$, then $q_f = q_g = q$ and there
    is a domain-weighted isomorphism $\sigma: [q] \to [q]$ from $(\fc,\alpha)$ to $(\gc,\beta)$ such that
    $\psi = \sigma \circ \varphi$.
\end{lemma}
\begin{proof}
Consider $(n-1)k$-labeled variable set $V_1 = \{v\} \cup \{u_a^{(d)}\}_{a \in [k], d \in [n-1]}$, 
where $u^{(d)}_a$ is labeled $a+(d-1)k$.
For a matrix $\chi \in ([n] \cup \{\bot\})^{k \times t}$ satisfying $\chi(*,j) \subset [n_j] \cup \{\bot\}$,
define the following set of constraints on $V_1$:
\[
    C_{\chi} = \{(F_j, u_a^{(1)}, \ldots, u_a^{(r-1)}, v, u_a^{(r)}, \ldots, u_a^{(n_j)}) \mid a \in [k], j \in [t], r \in [n_j], \chi(a,j) = r\}.
\]
Add no constraints for $\chi(a,j) = \bot$.
Define a $(n-1)k$-labeled \#CSP$(\fc)$ instance $K_{\chi} = (V_1, C_{\chi}) \in \plisimp[\fc;(n-1)k]$.

Now construct a certain family of $\chi$. 
$|I_{\vx}| \geq 2nq_f^n \geq 2nq_fq_g^{n-1}$ for every $\vx \in [q_f]^{n-1}$, so, by the pigeonhole principle,
there is a function
$s: [q_f]^{n-1} \to [q_g]^{n-1}$ such that for every $\vx \in [q_f]^{n-1}$, there exists a
$J_{\vx} \subseteq I_{\vx}$ such that $|J_{\vx}| \geq 2nq_f$ and,
for every $a \in J_{\vx}$, $(\psi(a + (d-1)k))_{d \in [n-1]} = s(\vx)$. 
For $n_j \leq n$ and $\vx \in [q_f]^{n_j-1}$, let $\ext(\vx) \in [q_f]^{n-1}$ denote the extension of 
$\vx$, defined by 
\[
    \ext(\vx)_i = \begin{cases}
        x_i & i \leq n_j-1 \\ 
        1 \in [q_f] & i > n_j-1
    \end{cases}.
\]
The choice of $1 \in [q_f]$ is arbitrary. 
If $n_j = 1$ and $\vx = ()$, then $\ext(\vx)$ is the all-ones vector, though this is completely
arbitrary, as in this case $(F_j, u_a^{(1)}, \ldots, u_a^{(r-1)}, v, u_a^{(r)}, \ldots, u_a^{(n_j)})
= (F_j, v)$ and $F_j(x_1^{r-1},i,x_r^{n_j-1}) = F_j(i)$, so one will see below that the entries of 
$\ext(\vx)$ are irrelevant.

Extend $s$ to a function on $\bigcup_{d=1}^{n-1} [q_f]^d$ by
$s(\vx) := s(\ext(\vx))$.
For every $(j,\vx,r) \in \jc(\fc)$ \eqref{eq:jc},
fix an arbitrary subset $P_{j \vx r} \subset J_{\ext(\vx)}$
with cardinality $p_{j \vx r}$, such that, for fixed $\vx$ and $j$, $P_{j \vx r}$ are disjoint for
distinct $r \in [n_j]$. This is possible for values of $p_{j \vx r}$ up to $2q_f$ 
because $|J_{\ext(\vx)}| \geq 2nq_f$ and $r$ can take at most $n$ distinct values.
For a fixed choice of $(p_{j \vx r}) \in [0,2q_f)^{\jc(\fc)}$,
define $\chi = \chi((P_{j \vx r})_{(j,\vx,r) \in \jc(\fc)})$ by
$\chi(a,j) = r$, for all $(j,\vx,r) \in \jc(\fc)$ and $a \in P_{j \vx r}$. Set the remaining entries of $\chi$ to $\bot$.
Let $R$ be the set of all such matrices $\chi$ for all choices of 
$(p_{j \vx r}) \in [0,2q_f)^{\jc(\fc)}$.

To recap, for $j \in [t]$, if $a \in P_{j \vx r} \subset J_{\ext(\vx)} \subset I_{\ext(\vx)}$ for $\vx \in [q_f]^{n_j-1}$ and
$r \in [n_j]$, then $\varphi(a + (d-1)k) = x_d$ and $\psi(a + (d-1)k) = s(\ext(\vx))_d = s(\vx)_d$ 
for $d \in [n_j-1]$. Hence the variable $u_a^{(d)}$ takes value $x_d$ and $s(\vx)_d$ under $\varphi$ and $\psi$,
respectively, for $d \in [n_j-1]$. These values are independent of the choice of $a$ within
$P_{j \vx r}$.
By construction, $K_{\chi}$ contains a constraint 
$(F_j, u_a^{(1)}, \ldots, u_a^{(r-1)}, v, u_a^{(r)}, \ldots, u_a^{(n_j)})$ for every $a \in P_{j \vx r}$.
Therefore $Z_{\fc, \alpha}^{\varphi}(K_{\chi}) = Z_{\gc,\beta}^{\psi}((K_{\chi})_{\fc\to\gc})$ for every $\chi \in R$
is equivalent to: for all $(p_{j \vx r}) \in [0,2q_f)^{\jc(\fc)}$,
\[
    \sum_{i=1}^{q_f} \alpha_i \prod_{(j,\vx,r) \in \jc(\fc)} F_{j}(x_1^{r-1}, i,x_r^{n_j-1})^{p_{j \vx r}}
    =
    \sum_{i=1}^{q_g} \beta_i \prod_{(j,\vx,r) \in \jc(\fc)} G_{j}(s(\vx)_1^{r-1},i,s(\vx)_r^{n_j-1})^{p_{j \vx r}},
\]
where the sum over $i$ corresponds to the choice of assignment for the only free variable $v$.
Subtracting the RHS, we are left with a sum of $q_f + q_g \leq 2q_f$ terms on the LHS. 
Treating $[q_f]$ and $[q_g]$ as disjoint, apply \autoref{lem:vandermonde} to this sum with 
\[
    I := [q_f] \sqcup [q_g],
    ~J := \jc(\fc),
    ~a_i := \begin{cases} \alpha_i & i \in [q_f] \\ \beta_i & i \in [q_g] \end{cases},
    ~b_{i,j \vx r} := \begin{cases} F_{j}(x_1^{r-1}, i,x_r^{n_j-1}) & i \in [q_f] \\
    G_{j}(s(\vx)_1^{r-1},i,s(\vx)_r^{n_j-1}) & i \in [q_g] \end{cases}.
\]

$\fc$ is twin-free, so the tuples $(F_j(x_1^{r-1},i,x_r^{n_j-1}))_{(j,\vx,r) \in \jc(\fc)}$ 
are pairwise
distinct for $i \in [q_f]$. Hence no equivalence class $I_{\ell}$ contains more than one element of
$[q_f]$. However, every $\alpha_i \neq 0$ by definition, so no equivalence class contains only a single
element of $[q_f]$. Thus there is a function $\sigma: [q_f] \to [q_g]$ such that every 
$i \in [q_f]$ in an equivalence class with
$\sigma(i) \in [q_g]$ -- that is
\begin{equation}
    (F_{j}(x_1^{r-1}, i,x_r^{n_j-1}))_{(j,\vx,r) \in \jc(\fc)} =
    (G_{j}(s(\vx)_1^{r-1},\sigma(i),s(\vx)_r^{n_j-1}))_{(j,\vx,r) \in \jc(\fc)}.
    \label{eq:beta2}
\end{equation}
Since no two elements of $[q_f]$ are in the same equivalence class, $\sigma$ is injective, hence
bijective, as $q_f \geq q_g$. Thus $q_f = q_g = q$, and we view $\sigma$ as a function $[q] \to [q]$.
\autoref{lem:vandermonde} then gives
\begin{equation}
    \alpha_i = \beta_{\sigma(i)} \text{ for } i \in [q].
    \label{eq:alpha2}
\end{equation}

Next, define another family of \#CSP$(\fc)$ instances. Fix $F \in \fc$ and corresponding $G \in \gc$,
with common arity $n_F$.
Define a new $(n-1)k$-labeled variable set
\[
    V_2 = \{v_h \mid h \in [n_F]\} \cup \{u_a^{(d)}\}_{a \in [k], d \in [n-1]}
\]
where $u^{(d)}_a$ is labeled $a+(d-1)k$.
For $\chi_1, \ldots, \chi_{n_F} \in R$, define the following set of constraints on $V_2$:
\begin{align*}
    C_{\chi_1^{n_F}} = &\{(F, v_1,\ldots,v_{n_F})\} \\
    &\cup \{(F_j, u_a^{(1)}, \ldots, u_a^{(r-1)}, v_h, u_a^{(r)}, \ldots, u_a^{(n_j)}) \mid a \in [k], j \in [t], r \in [n_j], h \in [n_F], \chi_h(a,j) = r\}.
\end{align*}
Let $K_{\chi_1^{n_F}} = (V_2, C_{\chi_1^{n_F}}) \in \plisimp[\fc;(n-1)k]$.
Every $\chi_h \in R$ corresponds to a sequence of subsets
$(P_{h, j \vx r})_{(j,\vx,r) \in \jc(\fc)}$
and sequence of integers $(p_{h, j \vx r})_{j \vx r} \in [0,2q)^{\jc(\fc)}$
such that
$p_{h, j \vx r}$ is the cardinality of
$P_{h, j \vx r}$,
$P_{h, j \vx r} \subset J_{\ext(\vx)} \subset I_{\ext(\vx)}$, and
$\chi_h(a,j) = r$ for $j \in [t]$, $\vx \in [q]^{n_j-1}$, and $a \in P_{h,j\vx r}$,
and is $\bot$ elsewhere.

Now the assumption $Z_{\fc, \alpha}^{\varphi}(K_{\chi_1^{n_F}}) = Z_{\gc,\beta}^{\psi}((K_{\chi_1^{n_F}})_{\fc\to\gc})$ for every 
$(\chi_1^{n_F}) \in R^{n_F}$
is equivalent to: for all $(p_{h, j\vx r})_{h,j\vx r} \in [0,2q)^{[n_F] \times \jc(\fc)}$,
\begin{align*}
    &\sum_{\vi \in [q]^{n_F}} \left(\prod_{h=1}^{n_F} {\alpha_{i_h}} \right)F(\vi)
    \prod_{(j,\vx,r) \in \jc(\fc), h \in [n_F]}
    F_{j}(x_1^{r-1}, i_h ,x_r^{n_j-1})^{p_{h, j \vx r}} \\
    = &\sum_{\vi \in [q]^{n_F}} \left(\prod_{h=1}^{n_F} {\beta_{i_h}} \right)G(\vi)
    \prod_{(j,\vx,r) \in \jc(\fc), h \in [n_F]}
    G_{j}(s(\vx)_1^{r-1}, i_h,s(\vx)_r^{n_j-1})^{p_{h, j \vx r}}
\end{align*}
where the sum over $\vi$ corresponds to the choice of assignment for the free variables $v_1^{n_F}$.
Subtracting the RHS and applying \eqref{eq:alpha2} and \eqref{eq:beta2} gives
\[
    \sum_{\vi \in [q]^{n_F}} \left(\prod_{h=1}^{n_F} {\alpha_{i_h}} \right)
    (F(\vi) - G(\sigma(\vi)))
    \prod_{(j,\vx,r) \in \jc(\fc), h \in [n_F]}
    F_{j}(x_1^{r-1}, i_h ,x_r^{n_j-1})^{p_{h, j \vx r}}
    = 0.
\]
$(F_j(x_1^{r-1},i,x_r^{n_j-1}))_{(j,\vx,r) \in \jc(\fc)}$ are distinct for $i \in [q]$,
so the tuples $(F_j(x_1^{r-1},i_h,x_r^{n_j-1}))_{h \in [n_F], (j,\vx,r) \in \jc(\fc)}$
are distinct for distinct $\vi \in [q]^{n_F}$. Applying \autoref{cor:vandermonde} with 
\begin{align*}
    &m := n_F,~ J = \jc(\fc),~
    a_{\vi} := \left(\prod_{h=1}^{n_F} {\alpha_{i_h}} \right) (F(\vi) - G(\sigma(\vi))), \text{ and }
    b_{i,j \vx r} := F_{j}(x_1^{r-1}, i ,x_r^{n_j-1}),
\end{align*}
we obtain $\left(\prod_{h=1}^{n_F} {\alpha_{i_h}} \right) (F(\vi) - G(\sigma(\vi))) = 0$ for all $\vi$ (each $p_{h, j\vx r}$ ranges over $[0,2q] \supset [0,q]$). Since each $\alpha_i \neq 0$ and our choice of $F$ and $G$
was arbitrary, this implies
\begin{equation}
    F(\vi) =  G(\sigma(\vi)) \text{ for every } \vi \in [q]^{n_F} \text{ and
        every corresponding } F \in \fc, G \in \gc.
    \label{eq:sigma2}
\end{equation}
Combined with \eqref{eq:alpha2}, \eqref{eq:sigma2} implies that $\sigma$ is a domain-weighted isomorphism 
between $(\fc,\alpha)$ and $(\gc,\beta)$.
Since $\fc$ is twin-free, \eqref{eq:sigma2} also implies $\gc$ is also twin-free.

It remains to show that $\psi = \sigma \circ \varphi$. Again let $F \in \fc$ and $G \in \gc$ be
corresponding constraint functions with common arity $n_F$.
Fix $c \in [(n-1)k]$. 
We aim to show that $F(\varphi(c),\vi) = G(\psi(c),\sigma(\vi))$ for all $\vi \in [q]^{n_F-1}$. 
If $n_F = 1$, let $K \in \plisimp[\fc;(n-1)k]$ be an instance with no unlabeled/free variables and a
single constraint $(F,v_c)$, where $v_c$ is the variable labeled $c$. Then by assumption we have
\begin{equation}
    F(\varphi(c)) = Z_{\fc,\alpha}^{\varphi}(K) = Z_{\gc,\beta}^{\psi}(K_{\fc\to\gc}) = G(\psi(c)),
    \label{eq:unary}
\end{equation}
as desired (recall $\vi \in [q]^0 \implies \vi = ()$).

Otherwise, if $n_F \geq 2$, define a third family of \#CSP$(\fc)$ instances as follows.
Define a $(n-1)k$-labeled variable set
\[
    V_3 = \{v_h \mid h \in [n_F-1]\} \cup \{u_a^{(d)}\}_{a \in [k], d \in [n-1]}
\]
where $u^{(d)}_a$ is labeled $a+(d-1)k$
(equivalent to the previous $V_2$, but we have removed the free variable $v_{n_F}$).
% Consider $\chi, \chi' \in R^2$, 
% but when choosing each $P_{\vx j} \subseteq J_{xy}$ to define $\chi$ and each
% $P'_{\vx j} \subseteq J_{xy}$ to define $\chi'$, choose these subsets so that
% $c \in I_{\widehat{x}\widehat{y}} \setminus (P_{\widehat{x}\widehat{y},j} \cup P'_{\widehat{x}\widehat{y},j})$ for every $j \in [t]$ (so $\chi_{c,j} = \chi'_{c,j} = 0$). This is always possible because
% $p_{\widehat{x}\widehat{y},j}, p'_{\widehat{x}\widehat{y},j} < 2q$ and $|I_{\widehat{x}\widehat{y}}| \geq 4m^3 \geq 4m$.
Write $c = a_c + (d_c-1)k$ (so that $u_{a_c}^{(d_c)}$ is labeled $c$).
For $\chi_1, \ldots, \chi_{n_F-1} \in R$, define the following set of constraints on $V_3$:
\begin{align*}
    C_{\chi_1^{n_F-1}} = &\{(F, u_{a_c}^{(d_c)}, v_1,\ldots,v_{n_F-1})\} \\
    &\cup \{(F_j, u_a^{(1)}, \ldots, u_a^{(r-1)}, v_h, u_a^{(r)}, \ldots, u_a^{(n_j)}) \mid a \in [k], j \in [t], r \in [n_j], h \in [n_F-1], \chi_h(a,j) = r\}.
\end{align*}
Let $K_{\chi_1^{n_F-1}} = (V_3, C_{\chi_1^{n_F-1}}) \in \plisimp[\fc;(n-1)k]$.

Now the assumption $Z_{\fc, \alpha}^{\varphi}(K_{\chi_1^{n_F-1}}) = Z_{\gc,\beta}^{\psi}((K_{\chi_1^{n_F-1}})_{\fc\to\gc})$ for every 
$(\chi_1^{n_F-1}) \in R^{n_F-1}$
is equivalent to: for all $(p_{h, j\vx r})_{h,j\vx r} \in [0,2q)^{[n_F-1] \times \jc(\fc)}$,
\begin{align*}
    &\sum_{\vi \in [q]^{n_F-1}} \left(\prod_{h=1}^{n_F-1} {\alpha_{i_h}} \right) F(\varphi(c),\vi)
    \prod_{(j,\vx,r) \in \jc(\fc), h \in [n_F-1]}
    F_{j}(x_1^{r-1}, i_h ,x_r^{n_j-1})^{p_{h, j \vx r}} \\
    = &\sum_{\vi \in [q]^{n_F-1}} \left(\prod_{h=1}^{n_F-1} {\beta_{i_h}} \right) G(\psi(c), \vi)
    \prod_{(j,\vx,r) \in \jc(\fc), h \in [n_F-1]}
    G_{j}(s(\vx)_1^{r-1}, i_h,s(\vx)_r^{n_j-1})^{p_{h, j \vx r}}.
\end{align*}
Subtracting the RHS and applying \eqref{eq:alpha2} and \eqref{eq:beta2} gives
\[
    \sum_{\vi \in [q]^{n_F-1}} \left(\prod_{h=1}^{n_F-1} {\alpha_{i_h}} \right)
    (F(\varphi(c), \vi) - G(\psi(c), \sigma(\vi)))
    \prod_{(j,\vx,r) \in \jc(\fc), h \in [n_F-1]}
    F_{j}(x_1^{r-1}, i_h ,x_r^{n_j-1})^{p_{h, j \vx r}}
    = 0.
\]
% Subtracting the RHS and applying \eqref{eq:alpha2} and \eqref{eq:beta2} gives
% \[
%     \sum_{i,i'=1}^q \alpha_i \alpha_{i'} (F(\widehat{x},i,i') - G(\psi(c),\sigma(i),\sigma(i')))
%     \prod_{x,y,j}
%     F_j(i,x,y)^{p_{\vx j}}
%     F_j(i',x,y)^{p'_{\vx j}}.
% \]
As above, the tuples $(F_j(x_1^{r-1},i_h,x_r^{n_j-1}))_{h \in [n_F-1], (j,\vx,r) \in \jc(\fc)}$
are distinct for distinct $\vi \in [q]^{n_F-1}$. Hence by a similar application of
\autoref{cor:vandermonde} with $m := n_F-1$,
we have $F(\varphi(c),\vi) = G(\psi(c),\sigma(\vi))$ for all
$\vi \in [q]^{n_F-1}$. This holds for any corresponding pair $F \in \fc$ and $G \in \gc$. Additionally,
the reasoning is independent of the input order in 
$F(\varphi(c),\vi)$ and $G(\psi(c),\sigma(\vi))$. Hence, by \eqref{eq:beta2},
\[
    G_j(\sigma(\vi)_1^{r-1},\sigma(\varphi(c)),\sigma(\vi)_r^{n_j-1}) 
    = F_j(\vi_1^{r-1},\varphi(c),\vi_r^{n_j-1}) 
    = G_j(\sigma(\vi)_1^{r-1},\psi(c),\sigma(\vi)_r^{n_j-1}) 
\]
for all $(j,\vi,r) \in \jc(\gc)$.
Since $\gc$ is twin-free and $\sigma$ is a bijection, this gives $\sigma(\varphi(c)) = \psi(c)$.
We chose $c \in [(n-1)k]$ arbitrarily, so $\psi = \sigma \circ \varphi$.
\end{proof}

Now we remove the requirement that $\varphi$ be well-balanced, which in turn removes the requirement
that $k$ be very large. 
We also address the case where $\fc$ and $\gc$ contain only
unary constraint functions. 
\begin{theorem}
    \label{thm:nowellbalanced}
    Let $\fc$ and $\gc$ be compatible constraint function
    sets with domains $[q_f]$ and $[q_g]$, with $q_f \geq q_g$, and $\alpha$ and $\beta$ be the domain
    weights associated with $\fc$ and $\gc$, respectively. Assume $\fc$ is twin-free.
    Let $k \geq 0$ and $\varphi: [k] \to [q_f]$ and $\psi: [k] \to [q_g]$.
    If $Z_{\fc, \alpha}^{\varphi}(K) = Z_{\gc,\beta}^{\psi}(K_{\fc\to\gc})$ for every $K \in \plisimp[\fc;k]$, then 
    $q_f = q_g = q$ and there exists an domain-weighted isomorphism $\sigma: [q] \to [q]$ between
    $(\fc,\alpha)$ and $(\gc,\beta)$ such that
    $\psi = \sigma \circ \varphi$.
\end{theorem}
\begin{proof}
    First handle the case where $\fc$ and $\gc$ contain only unary constraint functions, where
    \autoref{lem:balanced} does not apply. Say $\fc = \gc = t$. For every $\vp \in [2q_f]^t$, let
    $K_{\vp} \in \plisimp[\fc;k]$ be the instance defined by ignoring the labeled variables, adding
    a single unlabeled variable $v$, and $p_j$ copies of the constraint $(F_j, v)$, for $j \in [t]$.
    Then $Z_{\fc,\alpha}^{\varphi}(K_{\vp}) = Z_{\gc,\beta}^{\psi}((K_{\vp})_{\fc\to\gc})$ 
    for every $\vp \in [2q_f]^t$ is equivalent to,
    for every $\vp \in [2q_f]^t$,
    \[
        \sum_{i=1}^{q_f} \alpha_i \prod_{j=1}^t F_j(i)^{p_j} -
        \sum_{i=1}^{q_g} \beta_i \prod_{j=1}^t G_j(i)^{p_j} = 0.
    \]
    Apply \autoref{lem:vandermonde} with $I := [q_f] \sqcup [q_g]$, $J := [t]$, $a_i := \alpha_i$ or $\beta_i$, and
    $b_{i,j} :=  F_j(i)$ or $G_j(i)$ for $i \in [q_f]$ or $i \in [q_g]$, respectively.
    $q_f \geq q_g$ and the tuples $(F_j(i))_{j \in [t]}$ are distinct for
    distinct $i \in [q_f]$ ($\fc$ is twin-free), so by similar reasoning to the first part of the proof
    of \autoref{lem:balanced}, $q_f = q_g = q$ and there is a domain-weighted isomorphism 
    $\sigma: [q_f] \to [q_g]$ between $(\fc,\alpha)$ and $(\gc,\beta)$. Now the unary function argument
    in the third step of the proof of \autoref{lem:balanced} concluding with \eqref{eq:unary}
    gives $F_j \circ \varphi = G_j \circ \psi$ for every $j \in [t]$.
    Hence $G_j \circ \sigma \circ \varphi = F_j \circ \varphi = G_j \circ \psi$ for every $j \in [t]$,
    so since $\gc$, being isomorphic to $\fc$, is twin-free, $\sigma \circ \varphi = \psi$.

    Otherwise, if $\fc$ and $\gc$ contain a function with arity $\geq 2$, the proof is a simple generalization of the proof of \cite[{Theorem 3.1}]{homomorphism}.
    We generalize $\ff$-weighted graphs $H$ and $H'$ to constraint function sets $\fc$ and $\gc$,
    $\ell$-labeled graphs $G \in \mathcal{PLG}^{\text{simp}}[\ell]$ to
    $\ell$-labeled \#CSP$(\fc)$ instances $K \in \plisimp[\ell]$, and $\text{hom}_{\mu}(G,H)$
    and $\text{hom}_{\nu}(G,H')$ to $Z^{\mu}_{\fc,\alpha}(K)$ and $Z^{\nu}_{\gc,\beta}(K_{\fc\to\gc})$,
    respectively.
    The \#CSP generalizations satisfy analogous properties to the special case of graph homomorphisms.
    In particular we use our \autoref{lem:balanced} in place of \cite[{Lemma 6.1}]{homomorphism}
    (with ``well-balanced'' in place of ``super-surjective'')
    and our \eqref{eq:mult} -- the multiplicativity of $Z_{\fc,\alpha}^{\psi}$ -- in place of the
    multiplicativity of $\hom_{\psi}(\cdot,H)$.
\end{proof}

Next, we introduce domain weights to constraint function sets with unit domain weights (equivalently, no domain weights) to 
remove the twin-free requirement. We have the following
generalization of \cite[{Corollary 6.2}]{homomorphism}
\begin{corollary}
    \label{cor:twins}
    Let $\fc$ and $\gc$ be compatible constraint function sets
    with domains $[q_f]$ and $[q_g]$.
    Let $k \geq 0$, $\varphi: [k] \to [q_f]$, and $\psi: [k] \to [q_g]$. If 
    $Z_{\fc}^{\varphi}(K) = Z_{\gc}^{\psi}(K_{\fc\to\gc})$ for every $K \in \plisimp[\fc;k]$, then
    $q_f = q_g = q$ and there is an isomorphism $\sigma: [q] \to [q]$ between $\fc$ and $\gc$ such
    that $\psi' = \sigma \circ \varphi$, where $\psi'(i)$ is a twin of $\psi(i)$ for every
    $i \in [k]$.
\end{corollary}
\begin{proof}
    The constraint function set twin-contraction procedure $\fc \mapsto \widetilde{\fc}$ described in
    \autoref{sec:preliminaries} is a generalization of and satisfies the same properties as the
    $\ff$-weighted graph contraction $H \mapsto \widetilde{H}$ in \cite{homomorphism}.
    Hence the proof is again a simple generalization of the proof of \cite[{Corollary 6.2}]{homomorphism},
    where we use \autoref{thm:nowellbalanced}, $\fc$, $\gc$, $[q_f]$, $[q_g]$, and $\plisimp[\fc;k]$
    in place of \cite[{Theorem 3.1}]{homomorphism}, $H$, $H'$, $V(H)$, $V(H')$, and
    $\mathcal{PLG}^{\text{simp}}[k]$, respectively.
\end{proof}

Finally, we have the following result for ordinary (unlabeled) \#CSP instances, a slightly stronger
version of \autoref{thm:mainresult}. Say an unlabeled 
\#CSP$(\fc)$ instance $K$ is \emph{simple} if $K \in \plisimp[\fc;0]$ (equivalently, the corresponding
bipartite variable-constraint incidince graph has no multiedges and the multiplicity of every
constraint in $C$ is 1 up to permutation of its variable order).
\begin{corollary}
    \label{cor:mainresult}
    Let $\fc$ and $\gc$ be compatible constraint function sets. Then $\fc \cong \gc$ if and only if
    $Z_{\fc}(K) = Z_{\gc}(K_{\fc\to\gc})$
    for every simple \#CSP$(\fc)$ instance $K$.
\end{corollary}
\begin{proof}
    We only need the backward direction, which is the $k=0$ case of \autoref{cor:twins}.
\end{proof}

The next observation is a generalization of \cite[{Remark 2}]{homomorphism}.
\begin{remark}
    \label{rem:constructive}
    For compatible constraint function sets $\fc$ and $\gc$ with common domain $[q]$
    and $\varphi, \psi: [k] \to [q]$, \autoref{thm:nowellbalanced}
    asserts that if there is no isomorphism $\sigma$ between $\fc$ and $\gc$ satisfying 
    $\psi = \sigma \circ \varphi$, then there is some witness instance $K \in \plisimp[\fc;k]$ such that
    $Z^{\varphi}_{\fc}(K) \neq Z^{\psi}_{\gc}(K_{\fc\to\gc})$. The proofs of 
    \autoref{lem:balanced} and \autoref{thm:nowellbalanced} provide an explicit finite list of
    instances in $\plisimp[\fc;k]$ guaranteed to contain such a witness. This finite list is constructed
    as follows. 
    In the proof of \autoref{thm:nowellbalanced} (see \cite{homomorphism}), we
    extend $\varphi$ to a well-balanced map with domain $[\ell]$, where 
    $\ell \leq k + 2nq^{2n-1}$ ($n$ is the maximum arity among functions in $\fc$). Let
    \begin{align*}
        S = \{(V_1,C_{\chi}) \mid \chi \in R\} 
        &\cup \{(V_2,C_{\chi_1^{n_F}}) \mid F \in \fc, \chi_1^{n_F} \in R\} \\
        &\cup \{(V_3,C_{\chi_1^{n_F-1}}) \mid F \in \fc, \chi_1^{n_F-1} \in R\}
        \subset \plisimp[\fc;\ell]
    \end{align*}
    be the (finite) set of all \#CSP$(\fc)$ instances constructed in the three steps of the proof of
    \autoref{lem:balanced}. 
    The proof of \autoref{thm:nowellbalanced} constructs the finite set 
    \[
        \mathcal{P} = \left\{\prod_{K \in S} K^{h_K} \mid \text{each } 0 \leq h_K < 2q^{\ell}\right\}
        \subset \plisimp[\fc;\ell],
    \]
    where $K^{h_K}$ is a product in $\plisimp[\fc;\ell]$. The proof shows that there is an isomorphism
    $\sigma$ between $\fc$ and $\gc$ satisfying $\psi = \sigma \circ \varphi$ if and only if
    $Z_{\fc}^{\varphi}(K) = Z_{\gc}^{\psi}(K_{\fc\to\gc})$ for every 
    $K \in \pi_{[k]}(\mathcal{P}) \subset \plisimp[\fc;k]$, where $\pi_{[k]}: \plisimp[\fc;\ell] \to \plisimp[\fc;k]$
    erases the labels $k+1,\ldots,\ell$.
\end{remark}

%% file: quantum_proof.tex
\section{The Intertwiner Proof}
\label{sec:quantum}

In this section, we give an alternate proof of \autoref{thm:mainresult} for the case $\ff = \mathbb{C}$.
Throughout this section, we also assume constraint function sets $\fc$ over $\c$ are
\emph{conjugate closed}, meaning $F \in \fc \iff \overline{F} \in \fc$, where $\overline{F}$ is the
entrywise conjugate of $F$. In particular, any set of real-valued constraint functions is
conjugate-closed.

The following construction `flattens' a constraint function into a matrix.
\begin{definition}[$F^{m,d}$, $f$]
    For constraint function $F$ of domain $[q]$ and arity $n$ and any $m,d \geq 0, m+d = n$, define
    $F^{m,d} \in \ff^{q^m \times q^d}$ by $F^{m,d}_{x_1\ldots x_m,x_n \ldots x_{m+1}} = F(\vx)$,
    where $x_1 \ldots x_m \in \mathbb{N}$ is the base-$q$ integer with most significant digit
    $x_1$, and similarly for $x_n \ldots x_{m+1}$. Abbreviate $f := F^{n,0} \in \ff^{q^n}$, called the
    \emph{signature vector} of $F$.
\end{definition}
Note that the bits of the column index of $F^{m,d}$ are reversed. This is done so that the definition
matches \autoref{def:sigmatrix} of a gadget signature matrix below.

\subsection{Holant Problems and Gadgets}

The proof is carried out in the \emph{Holant} framework, a generalization of \#CSP.
Like a \#CSP problem, a Holant problem $\holant(\mathcal{F})$ is
is parameterized by a set $\mathcal{F}$ of constraint functions, all on the same domain $V(\fc)$,
called \emph{signature functions} or \emph{signatures}.
The input to $\holant(\mathcal{F})$ is a \emph{signature grid} 
$\Omega$, which consists of an underlying multigraph with vertex set $V$ and edge set $E$, along with an
assignment to each $v \in V$ a signature $F_v \in \mathcal{F}$ of arity $\deg(v)$.
The incident edges $E(v)$ to $v$ are given an order and serve as the input variables to $F_v$, 
taking values in $V(\fc)$.
The output on input $\Omega$ is
\begin{equation}\label{eqn:def-holant}
    \holant_{\Omega}(\mathcal{F}) = \sum_{\sigma: E \to V(\fc)} \prod_{v \in V} F_v(\sigma |_{E(v)}),
\end{equation}
where $F_v(\sigma |_{E(v)})$ is the evaluation of $F_v$ on the ordered tuple $\sigma |_{E(v)}$, the
restriction of $\sigma$ to $E(v)$.
For example, counting perfect matchings or proper edge colorings are expressed
is expressed by assigning the {\sc Exact-One} or {\sc Disequality} function to each vertex, respectively.
For sets $\fc$ and $\gc$ of signatures define the problem $\holant(\fc \,|\, \gc)$,
which takes as input a signature grid with a bipartite underlying
multigraph with bipartition $V = V_1 \cup V_2$ such that the vertices in $V_1$ and $V_2$ are assigned
signatures from $\mathcal{F}$ and $\gc$, respectively.
% Holant is a more expressive framework than \#CSP. 
% We can express \#CSP in the Holant framework; the reverse does not hold in general~\cite{lovasz,cai-govorov}. 

We next define some particular constraint functions that we will use throughout this section.
\begin{definition}[$E_n, E^{m,d}$, $\eq$]
    \label{def:e}
    Define the $n$-ary \emph{equality} constraint function $E_n$ by $E_n(x_1,\ldots,x_n) = 1$ if $x_1 = \ldots = x_n$, and 0 otherwise.
    Write $E^{m,d} := (E_n)^{m,d}$, as we must have $m+d = n$.
    Define $\eq = \bigcup_n E_n$. 
\end{definition}

Let $\fc$ be a set of constraint
functions. To each \#CSP$(\fc)$ instance $K = (V,C)$ we associate a signature grid $\Omega_K$ in the context of
$\holant(\fc \,|\,  \mathcal{EQ})$ defined as follows: 
For every constraint $c \in C$, if $c$ applies function $F$ of arity $n$, create a degree-$n$ vertex 
$u_c$ assigned $F$, called a \emph{constraint vertex}.
For each variable $v \in V$, if $v$ appears in constraints $C_v \subseteq C$, 
create a degree-$|C_v|$ vertex $u_v$, called an \emph{equality vertex}, assigned $E_{|C_v|}$, and edges $(u_v, u_c)$ for every $c \in C_v$
(if $v$ appears in no constraints, the corresponding vertex is isolated). Assign the order of edges incident
to $u_c$ to match the order of variables in $c$.
Any edge assignment $\sigma$ must assign all edges incident to an equality vertex the same 
value (or else the term corresponding to $\sigma$ is 0), so we can view $\sigma$ as \#CSP variable assignent.
Hence $Z_{\fc}(K) = \holant_{\Omega_K}(\fc \mid \mathcal{EQ})$.

\begin{definition}[Gadget, $T(\k)$, $\gk(k,\ell)$, $\gk_{\fc}(k,\ell)$, $\gk_{\fc}$]
    \label{def:sigmatrix}
    A \emph{gadget} is a Holant signature grid equipped with an ordered set of dangling edges (edges with only one endpoint), defining external variables.

    Let $\k$ be a gadget with $n$ dangling edges and containing signatures of domain size $q$. 
    For any $k,\ell \geq 0$, $\ell + k = n$,
    define $\k$'s $(k,\ell)$-\emph{signature matrix}
    $T(\k) \in \c^{q^k \times q^{\ell}}$ by setting $T(\k)_{\vx,\vy}$ to be the Holant value
    when the first $k$ dangling edges (called \emph{output} dangling edges) are assigned $x_1,\ldots,x_k$ and the last $\ell$
    dangling edges (called \emph{input} dangling edges) are assigned $y_{\ell},\ldots,y_1$. Draw the output/input
    dangling edges to the left/right of the gadget, respectively, in cyclic order (outputs from top to
    bottom and inputs from bottom to top).

    Let $\gk(k,\ell)$ be the collection of all gadgets with $k$ output and $\ell$ input dangling edges,
    and $\gk_{\fc}(k,\ell) \subset \gk(k,\ell)$ be the subcollection of gadgets in the context of
    $\holant(\fc \mid \eq)$ and with all dangling edges incident to equality vertices.
    Let $\gk_{\fc} = \bigcup_{k,\ell} \gk_{\fc}(k,\ell)$.
\end{definition}
$\k$'s input dangling edges receive their inputs in reverse order in the definition of
$T(\k)$. This is done so that the output and input dangling edges both receive their inputs in order
from top to bottom, so that the dangling edges merged in the composition operation below line up
when we draw the gadgets being composed.
\begin{definition}[Gadget $\circ, \otimes, *$]
    \label{def:gadgetops}
    \quad
    \begin{itemize}
        \item Given $\k_1 \in \gk(j,k), \k_2 \in \gk(k,\ell)$, define the composition $\k_1 \circ \k_2
        \in \gk(j,\ell)$ by placing $\k_2$ to the right of $\k_1$, and merging the $i$th
        input dangling edge of $\k_1$ with the $k-(i-1)$st output dangling edge of $\k_2$, for $i \in [k]$.
        If composition makes vertices assigned 
        $E_a, E_b \in \eq$ adjacent, contract the edge between them and assign the resulting merged vertex
        $E_{a+b-2}$. This does not change the Holant value.
        \item For gadgets $\k_1 \in \gk(k_1,\ell_1), \k_2 \in \gk(k_2,\ell_2)$, define the tensor product
        $\k_1 \otimes \k_2 \in \gk(k_1+k_2,\ell_1 + \ell_2)$ by
        taking the disjoint union of the multigraphs underlying $\k_1$ and $\k_2$, placing
        $\k_1$ above $\k_2$.
        
        \item For $\k \in \gk(k,\ell)$, define the conjugate transpose $\k^* \in \gk(k, \ell)$ by reflecting
            $\k$s underlying multigraph horizontally, and replacing every signature $F$ with $\overline{F}$.
    \end{itemize}
\end{definition}
It is well known applying the $\circ, \otimes, *$ operations to gadgets corresponds to applying these operations to their
signature matrices. See e.g. \cite{cai_chen_2017}.

A $(k+\ell)$-labeled \#CSP$(\fc)$ instance $K \in \pli[\fc;k+\ell]$
corresponds to a gadget $\k \in \gc_{\fc}(k,\ell)$ with dangling edges incident to the equality vertices
constructed from the labeled variables.
For a map $\psi: [k+\ell] \to V(\fc)$ assigning the labeled variables 
$x_1,\ldots,x_k,y_1,\ldots,y_{\ell}$ (i.e. $\psi([k+\ell]) := (\psi(1),\ldots,\psi(k+\ell)) = (x_1,\ldots,x_k,y_1,\ldots,y_{\ell})$), we have
$T_{\fc}(\k)_{\vx,\vy} = Z^{\psi}(K)$, since giving an equality vertex an input $x$ along a dangling
edge forces all of its adjacent edges to take value $x$, pinning the corresponding variable to $x$.

\begin{definition}[$\e^{m,d}, \ii, \f$]
    For $m,d \geq 0$, let $\e^{m,d}$ be the gadget consisting of a single vertex, assigned $E_{m+d}$,
    with $m$ output and $d$ input dangling edges. Define $\ii := \e^{1,1}$.

    For $n$-ary signature function $F$ let $\f$ be the gadget consisting of
    a degree-$n$ vertex assigned $F$ and $n$ output dangling edges, with the $i$th dangling edge
    serving as the $i$th input to $F$. See \autoref{fig:decompose} for illustrations.
\end{definition}
Since the signature $E_{m+d}$ is symmetric in the order of its inputs, we do not have to specify which
input to $E_{m+d}$ each dangling edge corresponds to. Observe that 
\[
    T(\e^{m,d}) = E^{m,d} \text{ and } T(\f) = F^{n_F,0} = f.
\]

\begin{definition}[$\s_{\sigma}$, $S_{\sigma}$, $\s$, $S$]
    For permutation $\sigma \in \sb_k$, let $\s_{\sigma} \in \gk(k,k)$ be the gadget formed from $\ii^{\otimes k}$ by
    permuting the dangling ends of the input dangling edges according to $\sigma$ -- that is, the
    $i$th output dangling edge is incident to the same $E_2$ vertex as the $\sigma(i)$th input dangling
    edge. Here we consider both input and output dangling edges in top-to-bottom order, as we do for
    signature matrices.

    Define a $2k$-ary constraint function $S_{\sigma}$ by 
    \[
        S_{\sigma}(x_1^k,y_1^k) = \begin{cases} 1 & x_i = y_{\sigma(i)} \text{ for all } i \in [k] \\
            0 & \text{otherwise} \end{cases}.
    \]
    Then we have $T(\s_{\sigma}) = S_{\sigma}^{k,k}$.
    We will make particular use of $\s_{(1\ 2)} \in \gk(2,2)$, so we abbreviate $\s := \s_{(1\ 2)}$
    and $S := S_{(1\ 2)}$ (so that $S^{2,2} = T(\s)$).
\end{definition}
We will compose $\s_{\sigma}$ with other gadgets to permute their dangling edges.
We generally treat a gadget containing $E_2$ vertices as equal to the gadget created by erasing these vertices
from the edge they lie on, since this has no effect on the Holant value. Hence we can also view
$\s_{\sigma}$ as a gadget composed solely of two-sided dangling edges, with the $i$th output and
$\sigma(i)$th input dangling edges being the same edge, or as a braid where we ignore the crossing order.
Indeed, analogous to generating the braid group by crossing adjacent strands, we can construct any $\s_{\sigma}$ using only $\ii$ and $\s$:
\begin{lemma}
    \label{lem:permute}
    For any $k \in \n$ and $\sigma \in S_k$, we have
    $\s_{\sigma} \in \tcwdn{\ii,\s}$.
\end{lemma}
\begin{proof}
    Decompose $\sigma$ into adjacent transpositions as 
    $\sigma = (a_1\ a_1+1)(a_2\ a_2+1) \ldots (a_s\ a_s+1)$.
    Then, since $\s$ swaps the position of adjacent dangling edges, we have
    \[
        \s_{\sigma} = \bigcirc_{i=1}^s (\ii^{\otimes a_i-1} \otimes \s \otimes \ii^{\otimes k - a_i - 1}).
    \]
    See \autoref{fig:permute} for an illustration.
\end{proof}

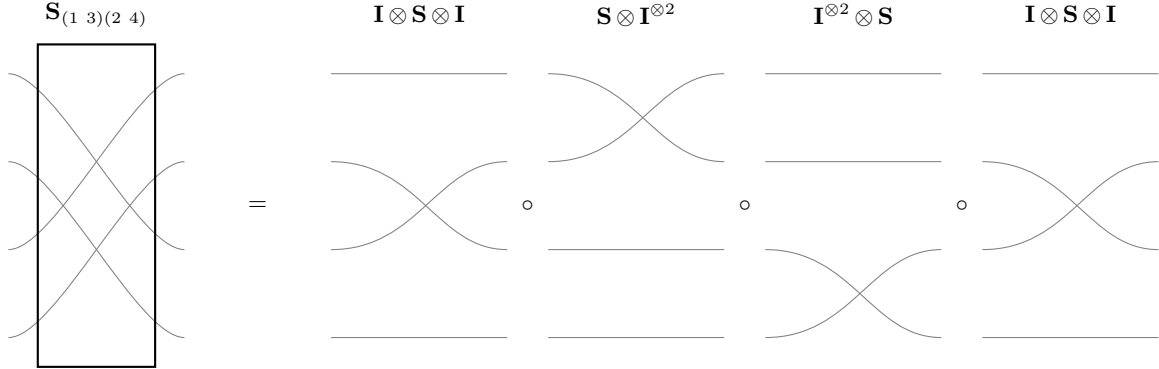
\begin{figure}[ht!]
    \center
    \input{permute.tex}
    \caption{Illustrating the decomposition of $\s_{\sigma}$ given by \autoref{lem:permute},
    with $\sigma = (1\ 3)(2\ 4) = (2\ 3)(1\ 2)(3\ 4)(2\ 3)$.}
    \label{fig:permute}
\end{figure}

\begin{theorem}
    \label{thm:generategk}
    For any conjugate-closed constraint function set $\fc$, 
    $\gk_{\fc} = \tcwdn{\e^{1,0},\e^{1,2},\s,\{\f \mid F \in \fc\}}$.
\end{theorem}
\begin{proof}
    Since $\fc$ is conjugate closed and $\e^{1,0}, \e^{1,2}$, and $\s$ are real-valued, the reverse inclusion $\supseteq$ is clear. To show the forward inclusion $\subseteq$, 
    first observe that $\ii = \e^{1,2} \circ (\e^{1,2})^* = \e^{1,2} \circ \e^{2,1}$,
    $\e^{m,d} = \bigcirc_{i=0}^{m-2} (\e^{2,1} \otimes \ii^{\otimes i}) \circ
    \bigcirc_{i=0}^{d-2} (\e^{1,2} \otimes \ii^{\otimes i})$
    for any $m,d \geq 2$, and
    $\e^{m,d} = \e^{m,1} \circ \bigcirc_{i=0}^{d-2} (\e^{1,2} \otimes \ii^{\otimes i})$
    for $m \in \{0,1\}$, $d \geq 2$. Also $\e^{0,0} = \e^{0,1} \circ \e^{1,0}$.
    Thus 
    \begin{equation}
        \label{eq:makeequality}
        \e^{m,d} \in \tcwdn{\e^{1,2},\e^{1,0}} \text{ for all } m,d \geq 0 
    \end{equation}
    (\eqref{eq:makeequality} is also a recontextualization of \cite[{Lemma 3.18}]{planar}). 

    Consider a
    $\holant(\fc \mid \eq)$ gadget $\k \in \gk_{\fc}(M,D)$. We will construct $\k$ from the
    fundamental gadgets. Suppose $\k$ contains $r$ equality vertices, 
    which we denote $e_1,\ldots,e_r$ in arbitrary order, and
    $s$ constraint vertices, denoted $c_1,\ldots,c_s$ in arbitrary order, with $c_j$ assigned signature
    $F_j \in \fc$.
    For $i \in [r]$, suppose vertex $e_i$ is
    incident to $m_i$ output and $d_i$ input dangling edges in $\k$, respectively, 
    and has degree $m_i + d_i + t_i$. Let $T = \sum_{i=1}^r t_i = \sum_{j=1}^s \deg(c_j)$, and we also
    have $M = \sum_{i=1}^r m_i$, $D = \sum_{i=1}^r d_i$, because by assumption all of $\k$'s dangling
    edges are incident to equality vertices.
    Let $\k_0 = \bigotimes_{i=1}^r \e^{m_i,d_i + t_i} \in \gk_{\fc}(M,D+T)$ and identify $e_i$ with
    the vertex in $\e^{m_i,d_i + t_i}$. By the bipartite structure of $\k$, 
    for all $k \in [\arity(F_1)] = [\deg(c_1)]$,
    the $k$th input edge to $c_1$ is incident to some equality vertex $e_{i_k}$,
    so $t_{i_k} > 0$. Thus, identifying $c_1$ with the vertex in $\f_1$, 
    there is a permutation $\sigma_1 \in \sb_{D+T}$ such that, in the gadget
    \[
        \k_1 = \k_0 \circ \s_{\sigma_1} \circ (\f_1 \otimes \ii^{\otimes D+T - \arity(F_1)}) \in \gk_{\fc}(M,D+T-\arity(F_1)),
    \]
    $c_1$'s $k$th incident edge is merged with the proper edge incident to $e_{i_k}$ for every
    $k \in [\arity(F_1)]$. 
    Similarly, for each $j \in [s]$, let $\sigma_j \in \sb_{D+T-\sum_{\ell=1}^{j-1} \arity(F_{\ell})}$ be the permutation that matches 
    $c_j$'s incident edges with the proper equality vertices. Then
    \[
        \k_s = \k_0 \circ \bigcirc_{j=1}^s \left(\s_{\sigma_j} \circ \big(\f_j \otimes \ii^{\otimes D+T - \sum_{\ell = 1}^{j}
        \arity(F_{\ell})}\big)\right)
        \in \gk_{\fc}(M,D),
    \]
    is a gadget with the same internal structure as $\k$, but with its input and output dangling edges
    permuted by some $\tau \in S_M$ and $\upsilon \in S_D$, respectively, relative to $\k$. Then
    $\k = \s_{\tau^{-1}} \circ \k_s \circ \s_{\upsilon}$.
    By \autoref{lem:permute} and \eqref{eq:makeequality}, we have 
    $\k \in \tcwdn{\e^{1,0},\e^{1,2},\s,\{\f \mid F \in \fc\}}$.
\end{proof}
See \autoref{fig:decompose} for an illustration. The proof of \autoref{thm:generategk} was inspired by
the proof sketch of \cite[{Theorem 8.4}]{planar}, which is roughly \autoref{thm:generategk} restricted
to unweighted graph homomorphism (the case where $\fc$ contains a single binary symmetric 0-1 valued constraint function).

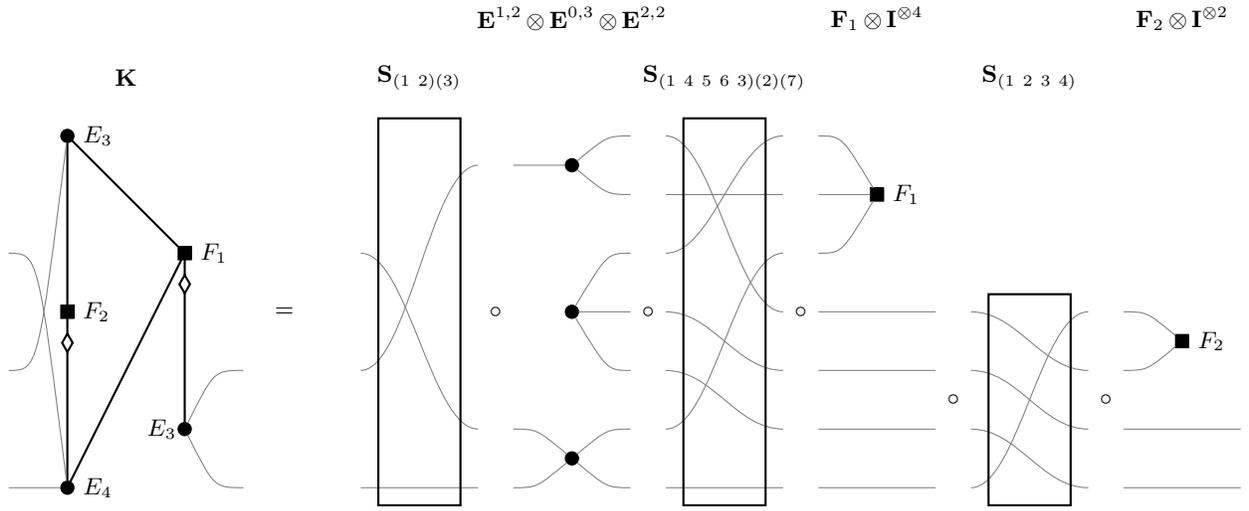
\begin{figure}[ht!]
    \input{decompose.tex}
    \caption{Illustrating the \autoref{thm:generategk} decomposition of a $\k \in \gk(3,2)$.
        We draw dangling edges thinner than internal edges, and use circles for equality vertices
        and squares for constraint vertices. A diamond on an edge marks this edge as the first input
    to the incident constraint vertex, and inputs proceed counterclockwise.}
    \label{fig:decompose}
\end{figure}

\subsection{Intertwiner Spaces}
Let $G$ be a subgroup of the symmetric group $\sb_q$. We identify elements $\sigma \in G$ with the associated permutation matrix $P_{\sigma}\in \{0,1\}^{q \times q}$.
The $(k,\ell)$-\emph{intertwiner space} of $G$ is
\[
    C_G(k,\ell) = \{T \in \c^{q^k} \times \c^{q^{\ell}} \mid
    \forall \sigma \in G: P_{\sigma}^{\otimes k} T = T P_{\sigma}^{\otimes \ell}\}.
\]
Define $C_G = \bigcup_{k,\ell} C_G(k,\ell)$ to be the space of all intertwiners of $G$.
If $P_{\sigma}$ is a $q \times q$ permutation matrix, then, for vector $\vv \in \mathbb{C}^{q^n}$, 
$P_{\sigma}^{\otimes n} v$ is the vector
obtained by permuting $v$'s entries according to the natural action of $\sigma$ on $[q]^n$
(the action $\sigma(\vx) = (\sigma(x_1),\ldots,\sigma(x_n))$).
Hence two indices in $[q^n]$ are in the same orbit of the action of $G$ if and only if every $(n,0)$
intertwiner takes equal values on the two indices -- that is, for $\vx,\vy \in [q^n]$,
\begin{equation}
    \label{eq:orbitvector}
    \text{There exists a } \sigma \in G \text{ such that } \sigma(\vx) = \vy \text{ if and only if }
    v_{\vx} = v_{\vy} \text{ for every } \vv \in C_G(n,0)
\end{equation}
(to see the reverse direction, suppose there is no such $\sigma$ and consider the $\vv$ which is 1 on
the orbit containing $\vx$ and 0 elsewhere).

It is well-known (see e.g. \cite{banica_liberation_2009}) that for any $G \subset \sb_q$, $C_G$ is a
\emph{symmetric tensor category with duals}, meaning each $C_G(k,\ell)$ is a vector space over $\c$ and
$C_G$ is closed under
matrix multiplication, tensor product, and conjugate transpose, and satisfies 
$I = E^{1,1} \in C_{G}(1,1)$,
$E^{2,0} \in C_{G}(2,0)$, and
$S^{2,2} \in C_G(2,2)$.

The next result is a version of classical Tannaka-Krein duality, proved by Woronowicz in 
\cite{woronowicz_tannaka}, and expressed in this form in \cite{chassaniol_study_2019},
\cite{banica_liberation_2009}, and elsewhere.
It is the key result underlying our alternate proof of \autoref{thm:mainresult}
\begin{theorem}
    The mapping $G \mapsto C_G$ induces a bijection between subgroups $G \subset S_q$ and symmetric
    tensor categories with duals $C$ satisfying $C_{S_q} \subset C$.
    \label{thm:tannaka}
\end{theorem}

The next lemma is an extension of \cite[{Proposition 3.5}]{chassaniol_study_2019}, which states that
for graph $X$ (equivalently, a symmetric binary 0-1 valued constraint function), 
$C_{\aut(X)} = \tcwd{E^{1,0},E^{1,2},S^{2,2},A_X}$ (the assumption in \cite{chassaniol_study_2019} that
$X$ is vertex-transitive is not necessary -- see also \cite[{Theorem 2.17}]{planar}). 
The proof in \cite{chassaniol_study_2019}, which makes implicit use of \autoref{thm:tannaka} to
obtain the forward inclusion $\subseteq$, easily
extends to a proof of the next lemma by generalizing the statement
$\forall \sigma \in \aut(X): P_{\sigma} A_X = A_X P_{\sigma}$ to
$\forall F \in \fc: \forall \sigma \in \aut(\fc): P_{\sigma}^{\otimes \arity(F)}f = f$,
which is equivalent to $\forall F \in \fc: f \in C_{\aut(\fc)}(\arity(F),0)$.
\begin{lemma}
    \label{thm:chassaniol}
    $C_{\aut(\fc)} = \tcwd{E^{1,0},E^{1,2},S^{2,2},\{f \mid F \in \fc\}}$.
\end{lemma}

Define a \emph{$(k,\ell)$-quantum $\fc$-gadget} to be a formal $\c$-linear combination of gadgets in $\gk_{\fc}(k,\ell)$.
In the context of graph homomorphism, where $\fc = \{F\}$, a binary constraint function, 
since a gadget in $\gc_{\{F\}}(k,\ell)$ corresponds to a $(k+\ell)$-labeled \#CSP$(\{F\})$ instance, 
a $(k,\ell)$-quantum $\{F\}$-gadget is equivalent to a \emph{$(k + \ell)$-labeled quantum graph} 
\cite{freedman_reflection_2006, lovasz, lovasz_contractors_2009}. 
Let $\qk_{\fc}(k,\ell)$ be collection of all $(k,\ell)$-quantum $\fc$-gadgets.
We extend the signature matrix function $T$ linearly to $\qk_{\fc}(k,\ell)$.
Observe that, for a fixed $(k,\ell)$, the set on the RHS of \autoref{thm:chassaniol} is the span of the signature
matrices of the gadgets in the set on the RHS of \autoref{thm:generategk}. Hence we have the following theorem.
\begin{theorem}
    \label{thm:intertwinersigmatrix}
    $C_{\aut(\fc)}(k,\ell) = \{T(\qb) \mid \qb \in \qk_{\fc}(k,\ell)\}$
    for every $k, \ell \in \mathbb{N}$.
\end{theorem}
By \eqref{eq:orbitvector} and the equivalence between $(k,\ell)$-quantum $F$-gadgets and
$(k + \ell)$-labeled quantum graphs,
the case $\ell = 0$ of \autoref{thm:intertwinersigmatrix} is a generalization (without domain weights)
from graph homomorphisms to \#CSP of \cite[{Lemma 2.5}]{lovasz} and \cite[{Theorem 9.3}]{homomorphism}
(the latter restricted to $\r$ rather than an arbitrary charactistic-0 field).

The next result is a similar generalization of \cite[{Lemma 2.4}]{lovasz}. It is also a version of \autoref{thm:nowellbalanced}
without domain weights and restricted to $\fc = \gc$ and $k > 0$, but without the twin-free assumption.
\begin{lemma}
    \label{lem:witness}
    Let $k > 0$ and $\varphi, \psi: [k] \to V(\fc)$. If $Z_{\fc}^{\varphi}(K) = Z_{\fc}^{\psi}(K)$
    for every $K \in \pli[\fc;k]$, then there is a $\sigma \in \aut(\fc)$ satisfying
    $\psi = \sigma \circ \varphi$.
\end{lemma}
\begin{proof}
    View $K$ as a gadget $\k \in \gk_{\fc}(k,0)$, so by assumption
    $(T(\k))_{\varphi([k])} = Z_{\fc}^{\varphi}(K) = 
    Z_{\fc}^{\psi}(K) = (T(\k))_{\psi([k])}$ for every $\k \in \gk_{\fc}(k,0)$, hence
    $(T(\qb))_{\varphi([k])} = (T(\qb))_{\psi([k])}$ for every $\qb \in \qk_{\fc}(k,0)$.
    Thus, by \autoref{thm:intertwinersigmatrix}, $v_{\varphi([k])} = v_{\psi([k])}$ for every
    $\vv \in C_{\aut(\fc)}(k,0)$, so by \eqref{eq:orbitvector}, there is a $\sigma \in \aut(\fc)$ satisfying
    $\sigma(\varphi([k])) = \psi([k])$. In other words, $\sigma \circ \varphi = \psi$.
\end{proof}

The final step is to use \autoref{lem:witness} to prove \autoref{thm:mainresult}
for $\ff = \c$ and CC $\fc$ and $\gc$.
We begin with the following definition.
\begin{definition}[$\oplus$]
    Let $F \in \ff^{V(F)^n}$, $G \in \ff^{V(G)^n}$ be constraint functions of arity $n > 1$, and assume $V(F)\cap V(G) = \varnothing$. For $n > 1$, the \emph{direct sum} $F\oplus G \in \ff^{(V(G)\cup V(F))^n}$ of $F$ and $G$ is defined by
    \[
        (F \oplus G)(\vx) = \begin{cases} F(\vx) & \vx \in V(F)^n \\ G(\vx) & \vx \in V(G)^n
        \\ 0 & \text{otherwise}\end{cases}
    \]
    for $\vx \in (V(G)\cup V(F))^n$. 
    For constraint function sets $\fc$ and $\gc$ of size $t$, define $\fc\oplus\gc = \{F_i \oplus G_i
    \mid i \in [t]\}$.
    \label{def:oplus}
\end{definition}
For $n = 2$, $F \oplus G$ is the adjacency matrix of the disjoint union of the $\ff$-weighted 
graphs with adjacency matrices $F$ and $G$. 

\begin{definition}[$\sim, \approx$, connected constraint function]
    \label{def:connected}
    For $n > 1$ and $F \in \ff^{V(F)^n}$, define an equivalence relation $\approx$ on $V(F)$ as the transitive closure 
    of the relation $\sim$, where $x \sim y$ if there is some tuple $\vx \in V(F)^n$ containing $x$ and
    $y$ such that $F(\vx) \neq 0$. Say $F$ is \emph{connected} if $\approx$ has exactly one
    equivalence class, and is \emph{disconnected} otherwise. 
\end{definition}

If $I, J \subseteq V(F)$ are distinct
equivalence classes of $\approx$ (`connected components') and $\sigma \in \aut(F)$ satisfies $\sigma(i) = j$ for some $i\in I$
and $j \in J$, it follows that $\sigma$ is an isomorphism between the subtensor of $F$ induced by $I$ and the
subtensor of $F$ induced by $J$. In particular, if $F$ and $G$ are connected, then $V(F)$ and $V(G)$ are
the two equivalence classes of $V(F \oplus G)$, so if a $\sigma \in \aut(F\oplus G)$ maps
some $x \in V(X)$ to some $y \in V(Y)$, then $F \cong G$. For $n = 2$ and symmetric $F$ and $G$, the
above statements are all equivalent to the corresponding well-known facts about graphs.

Now \autoref{lem:witness} gives another proof of our main result.
\begin{proof}[Proof of \autoref{thm:mainresult} for $\ff = \c$ and $\fc,\gc$ conjugate-closed]
    Assume $Z_{\fc}(K) = Z_{\gc}(K_{\fc\to\gc})$ for every \#CSP$(\fc)$ instance $K$.
    WLOG, we may assume $V(\fc)$ and $V(\gc)$ are disjoint.
    Let $0_F$ and $0_G$ be new domain elements. For each $F \in \fc$, $G \in \gc$ of arity $n \geq 2$,
    define constraint functions $F'$ and $G'$ on $V(\fc') := V(\fc) \cup \{0_F\}$ and $V(\gc') := V(\gc) \cup \{0_G\}$,
    respectively, by
    \[
        F'(\vx) = \begin{cases} F(\vx) & \vx \in V(\fc)^n \\ 1 & \text{otherwise} \end{cases},
        \qquad
        G'(\vx) = \begin{cases} G(\vx) & \vx \in V(\gc)^n \\ 1 & \text{otherwise} \end{cases}
    \]
    for $\vx \in V(\fc')^n$ and $\vx \in V(\gc')^n$, respectively. 
    In other words, if any entry of
    $\vx$ is $0_F$, then $F'(\vx) = 1$, and similarly for $G'$. 
    For unary $F \in \fc$, $G \in \gc$, define \emph{binary} $F'$ and $G'$ by
    \[
        F'(x,y) = \begin{cases} F(x) & x = y \in V(\fc) \\ 1 & x = 0_F \text{ or } y = 0_F \\
        0 & \text{otherwise}\end{cases},
        \qquad
        G'(x,y) = \begin{cases} G(x) & x = y \in V(\gc) \\ 1 & x = 0_G \text{ or } y = 0_G \\
        0 & \text{otherwise}\end{cases}
    \]
    for $x,y \in V(\fc')$ and $x,y \in V(\gc')$, respectively. The arity increase is necessary
    because the direct sum is only sensibly defined for constraint functions with arity $> 1$.
    Let $\fc' = \{F' \mid F \in \fc\}$ and $\gc' = \{G' \mid G \in \gc\}$.

    Let $K = (V,C) \in \pli[\fc' \oplus \gc';1]$ be a 1-labeled
    \#CSP$(\fc' \oplus \gc')$ instance, with labeled variable $v_0 \in V$. We will show that
    \begin{equation}
        \label{eq:zfzg}
        Z_{\fc'\oplus\gc'}^{0_F}(K) = Z_{\fc'\oplus\gc'}^{0_G}(K).
    \end{equation}
    If $K$ is not connected (i.e. the underlying graph of the $\holant(\fc \mid \eq)$ signature grid corresponding to $K$ is not connected), 
    then the components of $K$ that do not contain $v_0$ contribute the same
    value to the partition regardless of the assignment to $v_0$. Hence, to establish \eqref{eq:zfzg},
    we may assume $K$ is connected.
    Any variable assignment $\phi: V \to V(\fc' \oplus \gc')
    = V(\fc') \cup V(\gc')$ satisfying $\phi(v_0) = 0_F$ maps some $S \subset V$ to $0_F$, with
    $v_0 \in S$. 
    Furthermore, since $K$ is connected, $0_F \in V(\fc')$, and each $F' \oplus G'$ evaluates to 0
    unless all its inputs are in $V(\fc')$ or all its inputs are in $V(\gc')$, if such a $\phi$ makes
    a nonzero contribution to $Z^{0_F}_{\fc'\oplus\gc'}(K)$, we must have $\phi(V) \subset V(\fc')$.
    For a fixed $S \subset V$, the remaining variables $V \setminus S$ take all values in
    $V(\fc') \setminus \{0_F\} = V(\fc)$ as $\phi$ ranges over $\{\phi \mid \phi^{-1}(0_F) = S\}$.
    Additionally, any constraint containing a variable in $S$ always contributes 1, regardless of the
    assignments to the other variables.

    Construct a \#CSP$(\fc)$ instance $K^{\fc}_{V \setminus S}$ from $K$ as follows. First eliminate all
    variables in $S$ and all constraints containing any variable in $S$. Then, for each constraint
    applying $F' \oplus G'$, if $F$ and $G$ have arity $> 1$, replace $F'\oplus G'$ with $F$, and if
    $F$ and $G$ are unary, then merge the two variables to which the binary $F' \oplus G'$ is applied
    and replace the constraint with a constraint applying $F$ to the merged variable. 
    Assuming all inputs to $F' \oplus G'$ are in $V(\fc)$, this variable
    merging procedure does not change the value of the partition function, since by construction
    $F' \oplus G'$ acts as the function $(x,y) \mapsto \delta_{xy} F_x$. Now
    by the discussion in the previous paragraph, the contribution to $Z_{\fc'\oplus\gc'}^{0_F}(K)$ 
    of the assignments 
    $\phi$ satisfying $\phi^{-1}(0_F) = S$ is $Z_{\fc}(K^{\fc}_{V\setminus S})$. Thus
    \[
        Z_{\fc'\oplus\gc'}^{0_F}(K) = \sum_{S \subset V, S \ni v_0} Z(K_{V \setminus S}^{\fc}).
    \]
    A similar expression holds for $Z_{\fc'\oplus\gc'}^{0_G}(K)$, with
    $K_{V\setminus S}^{\gc} = (K_{V\setminus S}^{\fc})_{\fc\to\gc}$ in place of $K_{V\setminus S}^{\fc}$.
    Thus by assumption we have
    \[
        Z_{\fc'\oplus\gc'}^{0_F}(K) = \sum_{S \subset V, S \ni v_0} Z_{\fc}(K_{V \setminus S}^{\fc})
        = \sum_{S \subset V, S \ni v_0} Z_{\gc}(K_{V \setminus S}^{\gc}) = Z_{\fc'\oplus\gc'}^{0_G}(K),
    \]
    proving \eqref{eq:zfzg}. Now by \autoref{lem:witness} with $k = 1$, there is a
    $\sigma \in \aut(\fc' \oplus \gc')$ satisfying $\sigma(0_F) = 0_G$.
    Since the domain elements $0_F$ and $0_G$ satisfy $0_F \sim x$ for every $x \in V(\fc')$ 
    and $0_G \sim y$ for every $y \in V(\gc')$, $\fc'$ and $\gc'$ are connected. Hence by the discussion
    after \autoref{def:connected}, $\sigma \mid_{V(\fc')}$ is an isomorphism between $F'$ and
    $G'$ for every corresponding $F \in \fc$ and $G \in \gc$, so by construction and the fact that
    $\sigma(0_F) = 0_G$, $\sigma \mid_{V(\fc)}$ is an isomorphism between $F$ and $G$
    (if $F$ and $G$ are unary, $\sigma \mid_{V(\fc)}$ is really an isomorphism between the functions
    $(x,y) \mapsto \delta_{xy} F_y$ and $(x,y) \mapsto \delta_{xy} G_y$, but this implies an isomorphism
    between $F$ and $G$, since unary functions are isomorphic if and only if they have the same multiset
    of entries). Thus $\fc' \cong \gc'$.
\end{proof}
The proof of \autoref{thm:mainresult} is a generalization of Lov\'asz's proof of 
\cite[{Corollary 2.6}]{lovasz}, which is essentially \autoref{thm:mainresult} restricted to
real-weighted graphs ($\fc$ and $\gc$ contain a single binary constraint function). Both proofs use the
idea of adding a universal vertex to connect the graph/constraint function, since for weighted such objects
we cannot take the complement to assume connectedness.

\section{Discussion}
The interpolation proof is constructive (in the sense of \autoref{rem:constructive}), 
applies to any set of constraint functions over any
characteristic-0 field, and relies only on the simple idea in \autoref{lem:firstvandermonde}, but
requires a very detailed presentation.
The intertwiner proof has a cleaner presentation and demonstrates interesting new
connections between Holant and representation theory,
but only applies to CC constraint function sets over $\c$ and is nonconstructive.
In the proof of \autoref{lem:witness}, Tannaka-Krein duality 
(via \autoref{thm:intertwinersigmatrix}) guarantees the existence of a witness $K \in \pli[\fc;k]$
such that $Z_{\fc}^{\varphi}(K) \neq Z_{\fc}^{\psi}(K)$ if there is no $\sigma$ satisfying
$\sigma \circ \varphi = \psi$, but, unlike the
interpolation proof, does not provide an explicit finite list of instances that must contain $K$.
One desirable feature of a constructive proof, as discussed in \cite[{Section 7}]{homomorphism}, is
to make certain dichotomy theorems (e.g. \cite{cai2013graph}) \emph{effective}, meaning there is 
algorithm that decides whether the problem is \#P-hard (the dichotomy is decidable) and, if so, 
constructs a reduction from a \#P-hard problem (rather than simply asserting such a reduction exists).
Notably, the current complex-weighted \#CSP dichotomy \cite{cai-chen-complexity} is not even known to be
decidable; our results could someday be used in the proof of a decidable dichotomy.

%% file: permute.tex
\begin{tikzpicture}[scale=.78]
\tikzstyle{every node}=[font=\small]
\GraphInit[vstyle=Classic]
\SetUpEdge[style=-]
\SetVertexMath

\def\vgap{1.5} % horizontal distance between vertices
\def\ay{3*\vgap/2}
\def\by{\vgap/2}
\def\cy{-\vgap/2}
\def\dy{-3*\vgap/2}
\def\wid{3} % wire horizontal length
\def\rgap{0.5} % space between rectangle and dangling ends
\def\cgap{0.7} % horizontal space for \circ

% S_sigma
\def\slx{0} % start of left dangling edges
\draw[thin,color=gray](\slx,\ay)..controls (\slx+0.8,\ay) and (\slx+\wid-0.8,\cy)..(\slx+\wid,\cy);
\draw[thin,color=gray](\slx,\by)..controls (\slx+0.8,\by) and (\slx+\wid-0.8,\dy)..(\slx+\wid,\dy);
\draw[thin,color=gray](\slx,\cy)..controls (\slx+0.8,\cy) and (\slx+\wid-0.8,\ay)..(\slx+\wid,\ay);
\draw[thin,color=gray](\slx,\dy)..controls (\slx+0.8,\dy) and (\slx+\wid-0.8,\by)..(\slx+\wid,\by);
\draw[color=black,thick] (\slx+\rgap,\ay+\rgap) rectangle (\slx+\wid-\rgap,\dy-\rgap);

\def\six{\slx + \wid + 2.5} % x coordinate of first S \otimes I gadget
\foreach \x/\sy/\iya/\iyb in {
    \six/\by/\ay/\dy, \six+\wid+\cgap/\ay/\cy/\dy, \six+2*\wid+2*\cgap/\cy/\ay/\by, \six+3*\wid+3*\cgap/\by/\ay/\dy}
{
    \draw[thin, color=gray] (\x,\iya) -- (\x+\wid,\iya);
    \draw[thin, color=gray] (\x,\iyb) -- (\x+\wid,\iyb);
    \draw[thin, color=gray] (\x,\sy) .. controls (\x+\wid/2,\sy) and (\x+3*\wid/5,\sy-\vgap) .. (\x+\wid,\sy-\vgap);
    \draw[thin, color=gray] (\x,\sy-\vgap) .. controls (\x+\wid/2,\sy-\vgap) and (\x+3*\wid/5,\sy) .. (\x+\wid,\sy);
};

\node at (\slx + \wid + 1.25, 0) {$=$};
\node at (\six + \wid + \cgap/2, 0) {$\circ$};
\node at (\six + 2*\wid + 3*\cgap/2, 0) {$\circ$};
\node at (\six + 3*\wid + 5*\cgap/2, 0) {$\circ$};

\def\labely{\ay+1}
\node at (\slx+\wid/2,\labely) {$\s_{(1\ 3)(2\ 4)}$};
\node at (\six+\wid/2,\labely) {$\ii \otimes \s \otimes \ii$};
\node at (\six+3*\wid/2+\cgap,\labely) {$\s \otimes \ii^{\otimes 2}$};
\node at (\six+5*\wid/2+2*\cgap,\labely) {$\ii^{\otimes 2} \otimes \s$};
\node at (\six+7*\wid/2+3*\cgap,\labely) {$\ii \otimes \s \otimes \ii$};

\end{tikzpicture}

%% file: decompose.tex
\begin{tikzpicture}[scale=.78]
    \tikzstyle{every node}=[font=\small]
    \GraphInit[vstyle=Classic]
    \SetUpEdge[style=-]
    \SetVertexMath

    \def\vgap{1} % horizontal distance between wires
    % y coordinates of the 7 wires
    \def\ay{3*\vgap}
    \def\by{2*\vgap}
    \def\cy{\vgap}
    \def\dy{0}
    \def\ey{-\vgap}
    \def\fy{-2*\vgap}
    \def\gy{-3*\vgap}

    \def\wlen{1} % length of most dangling edges
    \def\cgap{0.6} % horizontal space for \circ
    \def\rgap{0.3} % space between rectangle and dangling ends

    \tikzset{VertexStyle/.style = {shape = rectangle, fill = black, minimum size = 5pt, inner sep=1pt, draw}}
    \Vertex[x=1,y=\dy,L={F_2}]{f2}
    \Vertex[x=3,y=\cy,L={F_1}]{f1}

    \tikzset{VertexStyle/.style = {shape = circle, fill = black, minimum size = 5pt, inner sep=1pt, draw}}
    \Vertex[x=1,y=\ay,L={E_3}]{e1}
    \Vertex[x=1,y=\gy,L={E_4}]{e2}
    \Vertex[x=3,y=\fy,L={E_3},Lpos=180,Ldist=-0.1cm]{e3}

    % Constructed gadget
    \DEdge{f2}{e2}
    \Edge(f2)(e1)
    \Edge(e1)(f1)
    \Edge(e2)(f1)
    \DEdge{f1}{e3}
    \draw[thin, color=gray] (0,\gy) -- (e2);
    \draw[thin, color=gray] (0,\cy) .. controls (0.5,\cy) .. (e2);
    \draw[thin, color=gray] (0,\ey) .. controls (0.5,\ey) .. (e1);
    \draw[thin, color=gray] (e3) .. controls (3.5,\ey) .. (4,\ey);
    \draw[thin, color=gray] (e3) .. controls (3.5,\gy) .. (4,\gy);

    \node at (4.7, 0) {$=$};

    % left S gadget to permute left dangling edges
    \def\slx{5.4 + \cgap} % start of left dangling edges
    \draw[thin,color=gray](\slx,\ey)..controls (\slx+0.8,\ey) and (\slx+2*\wlen-0.8,\ay/2+\by/2)..(\slx+2*\wlen,\ay/2+\by/2);
    \draw[thin,color=gray](\slx,\cy)..controls (\slx+0.8,\cy) and (\slx+2*\wlen-0.8,\fy)..(\slx+2*\wlen,\fy);
    \draw[thin,color=gray](\slx,\gy) -- (\slx+2*\wlen,\gy);
    \draw[color=black,thick] (\slx+\rgap,\ay+\rgap) rectangle (\slx+2*\wlen-\rgap,\gy-\rgap);

    % Initial tensor of equalities
    % x-coordinate of equality vertices
    \def\ex{\slx+3*\wlen+\cgap}
    \node at (\ex-\cgap/2-\wlen, 0) {$\circ$};

    \Vertex[x=\ex,y=\ay/2+\by/2,NoLabel]{e12}
    \Vertex[x=\ex,y=\dy,NoLabel]{e22}
    \Vertex[x=\ex,y=\fy/2+\gy/2,NoLabel]{e32}
    \draw[thin, color=gray] (\ex-\wlen,\ay/2+\by/2) -- (e12);
    \draw[thin, color=gray] (e12) .. controls (\ex+\wlen-0.4,\ay) .. (\ex+\wlen,\ay);
    \draw[thin, color=gray] (e12) .. controls (\ex+\wlen-0.4,\by) .. (\ex+\wlen,\by);

    \draw[thin, color=gray] (e22) -- (\ex+\wlen,\dy);
    \draw[thin, color=gray] (e22) .. controls (\ex+\wlen-0.4,\cy) .. (\ex+\wlen,\cy);
    \draw[thin, color=gray] (e22) .. controls (\ex+\wlen-0.4,\ey) .. (\ex+\wlen,\ey);

    \draw[thin, color=gray] (e32) .. controls (\ex+\wlen-0.4,\fy) .. (\ex+\wlen,\fy);
    \draw[thin, color=gray] (e32) .. controls (\ex+\wlen-0.4,\gy) .. (\ex+\wlen,\gy);
    \draw[thin, color=gray] (e32) .. controls (\ex-\wlen+0.4,\fy) .. (\ex-\wlen,\fy);
    \draw[thin, color=gray] (e32) .. controls (\ex-\wlen+0.4,\gy) .. (\ex-\wlen,\gy);

    % big S gadget
    \def\sx{\ex + \wlen + \cgap} % start of left dangling edges
    \node at (\sx-\cgap/2, 0) {$\circ$};

    \draw[thin,color=gray](\sx,\ay)..controls (\sx+0.8,\ay) and (\sx+2*\wlen-0.8,\dy)..(\sx+2*\wlen,\dy);
    \draw[thin,color=gray](\sx,\cy)..controls (\sx+0.8,\cy) and (\sx+2*\wlen-0.8,\ay)..(\sx+2*\wlen,\ay);
    \draw[thin,color=gray](\sx,\fy)..controls (\sx+0.8,\fy) and (\sx+2*\wlen-0.8,\cy)..(\sx+2*\wlen,\cy);
    \draw[thin,color=gray](\sx,\dy)..controls (\sx+0.8,\dy) and (\sx+2*\wlen-0.8,\ey)..(\sx+2*\wlen,\ey);
    \draw[thin,color=gray](\sx,\ey)..controls (\sx+0.8,\ey) and (\sx+2*\wlen-0.8,\fy)..(\sx+2*\wlen,\fy);
    \draw[thin,color=gray](\sx,\by) -- (\sx+2*\wlen,\by);
    \draw[thin,color=gray](\sx,\gy) -- (\sx+2*\wlen,\gy);
    \draw[color=black,thick] (\sx+\rgap,\ay+\rgap) rectangle (\sx+2*\wlen-\rgap,\gy-\rgap);

    % F_1 layer
    \def\fax{\sx+3*\wlen+\cgap} % vertex x-coordinate
    \node at (\fax-\cgap/2-\wlen, 0) {$\circ$};

    \tikzset{VertexStyle/.style = {shape = rectangle, fill = black, minimum size = 5pt, inner sep=1pt, draw}}
    \Vertex[x=\fax,y=\by,L={F_1}]{f12}
    \tikzset{VertexStyle/.style = {shape = circle, fill = black, minimum size = 5pt, inner sep=1pt, draw}}
    \draw[thin, color=gray] (\fax-\wlen,\ay) .. controls (\fax-\wlen+0.4,\ay) .. (f12);
    \draw[thin, color=gray] (\fax-\wlen,\by) -- (f12);
    \draw[thin, color=gray] (\fax-\wlen,\cy) .. controls (\fax-\wlen+0.4,\cy) .. (f12);

    \draw[thin, color=gray] (\fax-\wlen,\dy) -- (\fax+\wlen,\dy);
    \draw[thin, color=gray] (\fax-\wlen,\ey) -- (\fax+\wlen,\ey);
    \draw[thin, color=gray] (\fax-\wlen,\fy) -- (\fax+\wlen,\fy);
    \draw[thin, color=gray] (\fax-\wlen,\gy) -- (\fax+\wlen,\gy);

    % small S gadget
    \def\ssx{\fax + \wlen + \cgap} % start of left dangling edges
    \node at (\ssx-\cgap/2, \ey/2+\fy/2) {$\circ$};

    \draw[thin,color=gray](\ssx,\dy)..controls (\ssx+0.8,\dy) and (\ssx+2*\wlen-0.8,\ey)..(\ssx+2*\wlen,\ey);
    \draw[thin,color=gray](\ssx,\ey)..controls (\ssx+0.8,\ey) and (\ssx+2*\wlen-0.8,\fy)..(\ssx+2*\wlen,\fy);
    \draw[thin,color=gray](\ssx,\fy)..controls (\ssx+0.8,\fy) and (\ssx+2*\wlen-0.8,\gy)..(\ssx+2*\wlen,\gy);
    \draw[thin,color=gray](\ssx,\gy)..controls (\ssx+0.8,\gy) and (\ssx+2*\wlen-0.8,\dy)..(\ssx+2*\wlen,\dy);
    \draw[color=black,thick] (\ssx+\rgap,\dy+\rgap) rectangle (\ssx+2*\wlen-\rgap,\gy-\rgap);

    % F_2 layer
    \def\fbx{\ssx+3*\wlen+\cgap} % vertex x-coordinate
    \node at (\fbx-\cgap/2-\wlen, \ey/2+\fy/2) {$\circ$};

    \tikzset{VertexStyle/.style = {shape = rectangle, fill = black, minimum size = 5pt, inner sep=1pt, draw}}
    \Vertex[x=\fbx,y=\dy/2 + \ey/2,L={F_2}]{f22}
    \tikzset{VertexStyle/.style = {shape = circle, fill = black, minimum size = 5pt, inner sep=1pt, draw}}
    \draw[thin, color=gray] (\fbx-\wlen,\dy) .. controls (\fbx-\wlen+0.4,\dy) .. (f22);
    \draw[thin, color=gray] (\fbx-\wlen,\ey) .. controls (\fbx-\wlen+0.4,\ey) .. (f22);

    \draw[thin, color=gray] (\fbx-\wlen,\fy) -- (\fbx+\wlen,\fy);
    \draw[thin, color=gray] (\fbx-\wlen,\gy) -- (\fbx+\wlen,\gy);

    % Labels above gadgets
    \def\labely{\ay + 1}
    \node at (2,\labely) {$\k$};
    \node at (\slx+\wlen,\labely) {$\s_{(1\ 2)(3)}$};
    \node at (\ex,\labely+1) {$\e^{1,2} \otimes \e^{0,3} \otimes \e^{2,2}$};
    \node at (\sx+\wlen,\labely) {$\s_{(1\ 4\ 5\ 6\ 3)(2)(7)}$};
    \node at (\fax,\labely+1) {$\f_1 \otimes \ii^{\otimes 4}$};
    \node at (\ssx+\wlen,\labely) {$\s_{(1\ 2\ 3\ 4)}$};
    \node at (\fbx,\labely+1) {$\f_2 \otimes \ii^{\otimes 2}$};
\end{tikzpicture}